\documentclass[11pt]{article}
\usepackage[utf8]{inputenc} 
\usepackage[T1]{fontenc}   
\usepackage{textcomp}    
\usepackage[final]{graphicx}

\usepackage{tabularx} 
\usepackage{booktabs}

\usepackage{amsmath,amssymb,amscd}
\usepackage{amsthm,bbm}
\usepackage{mathtools}
\usepackage{float}
\usepackage[mathscr]{eucal}
\usepackage{color}
\usepackage{booktabs}
\usepackage{relsize}
\usepackage{subcaption} 
\usepackage{array} 
\usepackage[square,sort&compress,numbers]{natbib}
\usepackage[table]{xcolor}

\large\normalsize

\newtheorem{lemma}{Lemma}
\newtheorem{proposition}{Proposition}

\newtheorem{remark}{Remark}
\usepackage{comment}
\makeatletter
\@ifundefined{widetext}{}{}
\makeatother

\newcolumntype{L}[1]{>{\raggedright\arraybackslash}p{#1}}
\newcolumntype{C}[1]{>{\centering\arraybackslash}p{#1}}

\numberwithin{equation}{section}

\title{Multiplayer Games on Regular Graphs under Aspiration-Driven Imitation: Fixation Probabilities and the Dilution Principle
}
\author{Dhaker Kroumi$^1$\footnote{e-mail: dhaker.kroumi@kfupm.edu.sa} 
\\$^1$Department of Mathematics\\King Fahd University of Petroleum and Minerals\\Dhahran 31261, Saudi Arabia\\
 }
\date{}
\begin{document}
\maketitle


\begin{abstract}
We investigate how cooperation evolves in multiplayer games on regular graphs of degree \(k\) when individuals base their decisions on both their own aspiration levels and the success of their neighbors. We calculate explicit first-order formulas for fixation probabilities using pair approximation and diffusion methods under weak selection. These formulas separate aspiration effects from payoff effects filtered by the graph. We show that increasing the graph degree attenuates the selection effects driven by payoffs. In contrast, the contribution of the aspiration gap \(\Delta\alpha=\alpha_D-\alpha_C\) is independent of graph degree at first order. We apply the framework to public goods, stag-hunt, and snowdrift games and obtain corresponding conditions for the evolutionary success of cooperation. In the public goods and stag-hunt games, the conditions at large degree retain cost terms, so cooperation requires a sufficiently positive aspiration gap to compensate for the cost of contributing. In the snowdrift game, both benefit and cost contributions vanish in the locally tree-like regime at large degree, leaving the aspiration gap as the leading first-order contribution. Monte Carlo simulations on connected random regular graphs of degree \(k\) are consistent with the first-order predictions within their formal regime of validity, while simulations outside this regime are interpreted only as qualitative checks of the predicted trends.
\end{abstract}






\section{Introduction}
Understanding how cooperation persists when evolution favors individual selfishness is a fundamental challenge in evolutionary game theory. In traditional, well-mixed populations, individuals interact randomly. In these settings, cooperators pay a personal cost to generate benefits that are equally accessible to defectors. Because defectors enjoy the benefits without paying the costs, natural selection almost inevitably drives cooperation to extinction \cite{MS1982,A1984,HS1998,N2006,Nowak2006FiveRules}. However, real populations are rarely well mixed. Spatial structure and social networks fundamentally alter evolutionary outcomes because interactions are restricted to local neighborhoods. This local structure allows cooperators to form clusters, effectively insulating themselves from defectors and allowing mutual benefits to accumulate, a phenomenon known as network reciprocity. A large body of literature has explored how different network topologies and update rules either promote or inhibit this clustering, often relying on weak-selection theory and extensive simulation data \cite{NowakMay1992SpatialChaos,ST1998,HauertDoebeli2004SnowdriftInhibits,SantosPacheco2005ScaleFreeUnifying,LHN2005,OhtsukiHauertLiebermanNowak2006SimpleRule,SantosPachecoLenaerts2006HeterogeneousPNAS,SF2007,AllenNowak2014GamesOnGraphs,PercEtAl2017PhysRep,ALN2017}.

While network reciprocity is well understood for pairwise interactions, the theoretical landscape becomes much more complex when we consider multiplayer games on graphs. Many real social dilemmas involve groups rather than pairs. When placed on a network, these multiplayer interactions generate overlapping groups. An individual does not just participate in one game, but simultaneously influences and is influenced by multiple interacting neighborhoods. This overlap creates nonlinear payoff aggregation that is absent in models with two players \cite{GokhaleTraulsen2010PNAS,PenaWuArranzTraulsen2016,LiBroomDuWang2016PRE,WangPercSzolnoki2024NatCom}. Analyzing these dynamics requires moving beyond simple deterministic equations. In finite populations, evolutionary success is inherently stochastic. Therefore, success is best quantified using fixation probabilities, the likelihood that a single mutant strategy will eventually take over the entire population \cite{Nowak2004,Taylor2004finite,TNP2006}. These fixation probabilities are highly sensitive not only to the game's payoff structure but also to the microscopic update rules that dictate how individuals revise their strategies under the influence of local competition \cite{TraulsenPachecoNowak2007JTB,Tarnita2009,Sandholm2010Book,Nowak2010structured}.

More broadly, recent studies have shown that cooperation outcomes in structured populations are sensitive to the specific behavioral mechanism assumed: varying decision rules, dilemma strength, and ecological feedback can substantially reorganize or even reverse standard evolutionary predictions \cite{BegumIslamKabir2026CSFQlearning,BegumIslamKabir2026AMC,SultanaIslamKabir2026CSF,ShuvoKabir2024RSOS}.

Many models in evolutionary game theory assume that strategy updating is driven by relative payoff differences, often through imitation or selection biased toward more successful strategies \cite{HS1998,W1995,N2006,SF2007,TraulsenPachecoNowak2007JTB}. However, a complementary behavioral approach focuses on ``satisficing,'' a concept originating from Simon's bounded rationality \cite{S1955}. Under a rule driven by aspiration, or satisficing, individuals do not constantly scan the population for optimal strategies. Instead, they evaluate their current payoff against an internal benchmark, known as their aspiration level. As long as their payoff meets or exceeds this expectation, they maintain their current behavior. They only tend to change their strategy when they are dissatisfied. In standard models based on aspiration, an individual who fails to meet their aspiration level does not systematically seek out better alternatives; rather, they tend to adopt a new strategy selected at random \cite{KMRV1998,PPS1999,DWAW2014,DuWuWang2015SciRep,ZWDW2021}.

Social learning theory provides a behavioral motivation for the two-step update rule used here. Individuals who perform poorly or face uncertainty often rely more strongly on social information, a pattern described by state-contingent tendencies, such as copying when uncertain or copying when dissatisfied \cite{Kendal2018TiCS}. This motivates the first step of our rule: payoff below aspiration increases the probability of revision. However, dissatisfaction alone does not specify which alternative is chosen. Empirical and theoretical work on social learning also emphasizes the imitation of successful individuals, where individuals preferentially imitate demonstrators with higher payoff, status, or apparent competence \cite{Henrich_GilWhite_2001,Chudek_etal_2012,Kendal2018TiCS}. This motivates the second step of our rule: conditional on revision, the individual imitates a neighbor with probability proportional to that neighbor's payoff. Thus, the model translates a simple behavioral mechanism into a tractable evolutionary update rule: dissatisfaction makes revision more likely, and successful neighbors are more likely to be copied.

Neither standard imitation models nor traditional aspiration models capture this two-step logic. Motivated by these empirical observations, we propose a hybrid evolutionary rule: dissatisfaction determines whether an individual revises their strategy, while imitation biased by payoff of successful neighbors determines whose strategy they adopt. In our previous work \cite{K2021}, we introduced this joint dynamic of aspiration and imitation for games with two players, showing that an asymmetry in aspiration levels can reverse the evolutionary outcomes predicted by standard imitation models.

In this paper, we extend our framework to multiplayer games on regular graphs of degree \(k\), where individuals participate in overlapping groups of \(k+1\) players and payoffs are highly sensitive to local network composition. To handle these complex dynamics, we combine pair approximation \cite{MatsudaOgitaSasakiSato1992,vB2000} with a diffusion reduction under weak selection \cite{OhtsukiHauertLiebermanNowak2006SimpleRule} to derive explicit fixation probabilities and invasion criteria.

The present work lies at the intersection of two established lines of research. Dynamics of updating driven by payoffs have been studied in multiplayer games on regular graphs using powerful tools, including structure coefficients, pair approximation, diffusion arguments, and replicator reductions \cite{GokhaleTraulsen2010PNAS,LiBroomDuWang2016PRE,PenaWuArranzTraulsen2016,WangPercSzolnoki2024NatCom}. These approaches show how graph structure modifies selection based on payoffs, but they do not include revision gated by aspiration. Conversely, rules driven by aspiration model the decision to revise via self-comparison, where players compare their realized payoff against an internal benchmark and switch when dissatisfied \cite{KMRV1998,PPS1999,DWAW2014,DuWuWang2015SciRep,ZWDW2021}. Such rules can remove the usual graph effect on regular networks, but they do not include imitation of successful neighbors after dissatisfaction.

Our model combines these mechanisms. The resulting first-order fixation formula gives a decomposition specific to this hybrid rule in multiplayer graph interactions. The selection drift separates into an aspiration term independent of degree, proportional to \(\Delta\alpha=\alpha_D-\alpha_C\), and a game-specific payoff term filtered by structural coefficients depending on \(k\). We call the attenuation of this payoff term with increasing degree the dilution principle. Thus, the novelty is not simply that payoff effects decay with degree, but that this decay appears together with an aspiration component that is not diluted by connectivity at first order.

The remainder of the paper is organized as follows. Section~\ref{sec:model} introduces the mathematical model and the hybrid update rule. Sections~\ref{sec:public}, \ref{sec:staghunt}, and~\ref{sec:snowdrift} present the explicit analytical results for the public goods, stag-hunt, and snowdrift games, respectively. Section~\ref{sec:simulations} compares the first-order theory with Monte Carlo simulations in a representative setting under weak selection. Finally, Section~\ref{sec:discussion} concludes with a broader discussion of the theoretical implications.


\section{Model and analytical methodology}\label{sec:model}

Consider a finite population of $N$ individuals situated on a connected $k$-regular graph ($k \geq 3$), where each vertex represents an individual. Since the graph is regular, every individual $i$ is connected to exactly $k$ neighbors; we denote this set of neighbors as $\mathcal{N}(i)$, such that $|\mathcal{N}(i)| = k$. Individuals adopt one of two strategies: cooperation ($C$) or defection ($D$). Social interactions occur in groups of size $k+1$, where each group consists of one focal individual acting as the organizer and its $k$ direct neighbors as participants. Under this structure, every individual engages in exactly $k+1$ games per time step: once as the organizer of its own neighborhood, and $k$ times as a participant in the games organized by each of its $k$ neighbors.

In any given group of size $k+1$, payoffs are determined by the number of cooperators present. If a group contains $n$ cooperators and $k+1-n$ defectors, a cooperator receives a payoff $a_{n-1}$, while a defector receives $b_n$. The payoff structure is summarized by the matrix 
\[
\begin{array}{c|cccc}
   & 0 & 1 & \cdots & k \\ \hline
C & a_0 & a_1 & \cdots & a_k \\
D & b_0 & b_1 & \cdots & b_k \\
\end{array}
\]
After all games in a given time step are played, each individual $i$ accumulates an average realized payoff, defined as the average payoff earned from all $k+1$ interactions in which it participated.

 At each time step, after payoffs are realized, a randomly selected individual $i$ is given the opportunity to revise its strategy. It evaluates its current performance by comparing its average payoff $\pi_i$ to an intrinsic aspiration level $\alpha_i$, where $\alpha_i = \alpha_C$ if $i$ is a cooperator, and $\alpha_i = \alpha_D$ if $i$ is a defector. Following the principle of satisficing \cite{S1955}, individuals who meet their aspiration levels are reluctant to change, while those who are dissatisfied are prone to revision. The probability that individual $i$ decides to \emph{revise} its strategy is given by the Fermi function
\begin{equation}\label{eq:update_prob}
\Pr(i\; \text{revises})=\frac{1}{1+\exp\{\delta(\pi_i-\alpha_i)\}},
\end{equation}
where $\delta \ge 0$ represents the intensity of selection. This formulation is widely used in evolutionary game theory to model stochastic decision-making \cite{ST1998,TNP2006}. 

Once individual $i$ decides to revise its strategy (according to Eq.~\eqref{eq:update_prob}), the focal individual samples a neighbor $j\in \mathcal{N}(i)$ according to a multinomial logit (softmax) rule
\begin{equation}\label{eq:imitation_prob}
\Pr(i\ \text{imitates}\ j \mid i\ \text{revises})
=\frac{\exp(\delta \pi_j)}{\sum_{\ell\in \mathcal{N}(i)}\exp(\delta \pi_\ell)},
\end{equation}
and copies $j$'s strategy. This update rule drives a frequency-dependent evolutionary process, allowing highly successful strategies to spread more easily across the network. Altogether, these dynamics represent a two-step behavioral process. First, internal dissatisfaction increases the probability of revision. Second, imitation biased by payoffs makes successful neighbors more likely to be copied.

The updating mechanism in Eqs.~\eqref{eq:update_prob}--\eqref{eq:imitation_prob} is intended to represent a process of satisficing followed by search rather than unconditional payoff maximization. Following Simon's notion of \emph{satisficing} \cite{S1955}, an individual persists with its current strategy as long as realized outcomes meet an internal benchmark (aspiration level). When outcomes fall below this benchmark, the individual becomes more likely to revise, which we interpret as initiating an active search for alternatives. Conditional on revising, the individual then relies on \emph{social information} and preferentially copies successful neighbors. This structure, where dissatisfaction leads to copying biased by social information, is consistent with a large literature in social learning and cultural evolution showing that individuals tend to consult and use social information particularly when uncertain or performing poorly, and that copying is often biased toward individuals with high payoffs \cite{Rendell2010Science,Kendal2018TiCS}.

In this work, the small parameter $\delta$ serves a dual purpose: it governs
the sensitivity of the satisfaction gate in Eq.~\eqref{eq:update_prob} and
modulates the payoff bias during social imitation in Eq.~\eqref{eq:imitation_prob}.
Thus, the main model assumes a symmetric intensity, in which
aspiration responsiveness and imitation biased by payoffs operate on the same
weak-selection scale. This coupling is a modeling choice rather than a
mathematical necessity. A more general formulation with separate intensities $\delta_A>0$ and $\delta_I>0$ is discussed in the sections dedicated to each game and derived in Appendix~\ref{app:two_intensity}.

To study the evolutionary dynamics in the population, we employ the method of pair approximation. The graph is assumed to be locally tree-like with $k\ll N$, ensuring that short loops and clustering exert a negligible effect on the conditional neighborhood distributions. Within this sparse regular regime, adjacent vertices rarely share additional common neighbors, allowing local neighborhoods to be treated as tree-like to leading order. This constitutes the standard closure for random regular graphs. Let $p_C$ denote the global frequency of cooperators, and let $q_{C\mid C}$ be the conditional probability that the neighbor of a cooperator is also a cooperator; all other pair and conditional frequencies can be expressed in terms of $(p_C,q_{C\mid C})$. In the weak-selection regime $\delta\ll1$, the system operates on two distinct time scales: local correlations, captured by $q_{C|C}$, equilibrate rapidly on a fast $\mathcal{O}(1)$ time scale, whereas the global frequency $p_C$ evolves slowly on an $\mathcal{O}(\delta)$ time scale. Because $q_{C|C}$ quickly settles into a quasi-steady state determined entirely by $p_C$, we can approximate the system using a one-dimensional diffusion process given by
$$dp_C = \delta \mu(p_C) dt + \sqrt{\sigma^2(p_C)} dW_t,$$
with absorbing boundaries at $p_C=0$ (extinction of cooperation) and $p_C=1$ (fixation of cooperation). The infinitesimal mean $\mu$ combines two main effects, a term driven by the difference in aspiration levels $\Delta\alpha = \alpha_D - \alpha_C$, and a structural term $G$ capturing payoff differences under the pair-approximation closure. The infinitesimal variance $\sigma^2$ is driven only by random drift. See Appendix~\ref{AppendixA} for the mathematical details.

Let $\rho_C(1/N,\delta)$ denote the fixation probability of a single cooperator. This is the probability that one invading cooperator successfully generates a lineage that takes over a population of defectors. Similarly, let $\rho_D(1/N,\delta)$ be the probability that a single invading defector fixates within a population of cooperators. In the weak-selection regime with $N\delta\ll 1$, we have the following approximations to first order
\begin{subequations}\label{eq:rho_general}
\begin{align}
\rho_C\!\left(\frac{1}{N},\delta\right)&\approx\frac{1}{N}+\delta\left[\frac{\Delta\alpha}{4}+2\int_0^1 (1-y)\,G(y)\,dy\right], \label{eq:rhoC_general}\\
\rho_D\!\left(\frac{1}{N},\delta\right)&\approx\frac{1}{N}-\delta\left[\frac{\Delta\alpha}{4}+2\int_0^1 y\,G(y)\,dy\right]. \label{eq:rhoD_general}
\end{align}
\end{subequations}
Furthermore, the ratio of these fixation probabilities can be approximated as
\begin{equation}\label{eq:ratio_general}
\frac{\rho_C\!\left(\frac{1}{N},\delta\right)}{\rho_D\!\left(\frac{1}{N},\delta\right)}\approx1+N\delta\left(\frac{\Delta\alpha}{2}+2\int_0^1 G(y)\,dy\right).
\end{equation}
These formulas offer clear conditions for determining when network structure and aspiration differences shift evolutionary outcomes. All derivations leading from the microscopic update rule to the reduced diffusion are provided in Appendix~\ref{AppendixA}.

The weak-selection approximations in Eqs.~\eqref{eq:rho_general}--\eqref{eq:ratio_general} are obtained under the scaling assumptions $k\ll N$ and $N\delta\ll1$. The condition $k\ll N$ supports the pair-approximation description in terms of $p_C$ and $q_{C|C}$ on locally tree-like regular graphs. The condition $N\delta\ll1$ ensures that the first-order selection correction remains small relative to the neutral single-mutant fixation probability $1/N$. This point is especially important for the ratio in Eq.~\eqref{eq:ratio_general}, whose leading correction contains the product $N\delta$. This scaling issue is well known in weak-selection fixation calculations, where the order of the weak-selection and large-population limits can affect the resulting asymptotics \cite{SampleAllen2017Limits}.

Following standard practice in stochastic evolutionary dynamics in finite populations, we use three complementary criteria to quantify the evolutionary success of cooperation \cite{Nowak2004,Taylor2004finite,Tarnita2009,Nowak2010structured}. First, \emph{selection is said to favor the fixation of cooperation} if a single cooperator reaches fixation with a probability larger than that of neutral drift:
\[
\rho_C\!\left(\frac{1}{N},\delta\right)>\rho_C\!\left(\frac{1}{N},0\right)=\frac{1}{N}.
\]
This criterion measures whether selection increases the chance that cooperation can emerge from rarity.
Second, \emph{selection is said to disfavor the fixation of defection} if a single defector is less likely than neutral drift to take over a cooperative population:
\[
\rho_D\!\left(\frac{1}{N},\delta\right)<\rho_D\!\left(\frac{1}{N},0\right)=\frac{1}{N}.
\]
This criterion quantifies the resistance of cooperation against rare defector invasions. Because these first two conditions are independent in general stochastic dynamics, selection can theoretically favor or disfavor the fixation of both strategies simultaneously.
Thus, a third criterion is introduced by comparing the fixation probabilities directly. We say that \emph{selection favors cooperation over defection} if a cooperator is more likely to invade a population of defectors than a defector is to invade a population of cooperators:
\[
\rho_C\!\left(\frac{1}{N},\delta\right)>\rho_D\!\left(\frac{1}{N},\delta\right).
\]
This ``favored over'' criterion is widely used and connects naturally to long-run abundance in the presence of rare mutations \cite{Fudenbergimhof2006}. 
\begin{table}[H]
\centering
\caption{Summary of the three complementary criteria for the evolutionary success of cooperation.}
\label{tab:criteria}
\begin{tabular}{>{\raggedright\arraybackslash}p{0.29\textwidth}>{\centering\arraybackslash}p{0.16\textwidth}>{\raggedright\arraybackslash}p{0.43\textwidth}}
\toprule
Criterion & Condition & Interpretation \\
\midrule
Favors fixation of cooperation & $\rho_C > 1/N$ & Selection helps cooperation emerge from rarity \\
Disfavors fixation of defection & $\rho_D < 1/N$ & Cooperation resists rare defector invasion \\
Favors cooperation over defection & $\rho_C > \rho_D$ & Cooperation is favored relative to defection \\
\bottomrule
\end{tabular}
\end{table}

For the remainder of this paper, we simplify our notation by letting $\rho_C$ and $\rho_D$ denote the fixation probabilities of a single mutant, $\rho_C(1/N, \delta)$ and $\rho_D(1/N, \delta)$, respectively.


\section{Public goods game}\label{sec:public}
In this section, we apply our general analytical framework to the linear public goods game (PGG) \cite{hamburger1973nperson,fox1978public}. In this classic social dilemma, each cooperator incurs a personal cost $c$ to provide a benefit $b$ to the group. The total accumulated benefit is then distributed equally among all $k+1$ group members, regardless of whether they contributed. Consequently, a focal cooperator interacting in a group with $n$ other cooperators (yielding a total of $n+1$ contributors) receives a payoff
\begin{equation}
a_{n} = \frac{n+1}{k+1}b - c.
\end{equation}
Conversely, a focal defector in a group with $n$ cooperators contributes nothing, thereby avoiding the cost $c$, but still receives an equal share of the generated public good. Its payoff is
\begin{equation}
b_{n} = \frac{n}{k+1}b.
\end{equation}

By substituting these specific PGG payoffs into our diffusion approximations under weak selection, we can explicitly evaluate the evolutionary dynamics. The respective fixation probabilities for a single cooperator ($\rho_C$) and a single defector ($\rho_D$), along with their ratio, can be approximated as
\begin{subequations}\label{eq:fix_pgg_asp}
\begin{align}
\rho_C&\approx\frac{1}{N} + \frac{\delta}{4}\left( \Delta\alpha + \frac{3(k+3)}{(k+1)^2}b - 3c \right), \\
\rho_D&\approx\frac{1}{N} - \frac{\delta}{4}\left( \Delta\alpha + \frac{3(k+3)}{(k+1)^2}b - 3c \right), \\
\frac{\rho_C}{\rho_D}&\approx 1 + \frac{\delta N}{2} \left( \Delta\alpha + \frac{3(k+3)}{(k+1)^2}b - 3c \right).
\end{align}
\end{subequations}
The mathematical details of these derivations are provided in Appendix~\ref{AppendixB}.

A direct consequence of Eq.~\eqref{eq:fix_pgg_asp} is that the influence of the network structure is completely captured by the factor $\frac{k+3}{(k+1)^2}$. 
Because this factor strictly decreases and approaches zero as the degree \(k\) grows large, it illustrates the dilution effect.
As interaction groups become larger, the marginal benefit that a single cooperator provides to any given neighbor shrinks. Meanwhile, the personal cost of contributing, $c$, remains constant. Therefore, for a fixed set of parameters $(b, c, \Delta\alpha)$, increasing graph degree tends to reduce \(\rho_C\) and increase \(\rho_D\), driving $\rho_C$ down and $\rho_D$ up.

In contrast, the aspiration gap $\Delta\alpha = \alpha_D - \alpha_C$ acts as an additive shift that is entirely independent of the degree $k$. A larger $\Delta\alpha$ simultaneously boosts the fixation probability of cooperators and suppresses that of defectors. Behaviorally, a positive $\Delta\alpha$ means that defectors require higher payoffs to feel satisfied. Because they are harder to satisfy, defectors are more likely to revise their strategies and, through imitation biased by payoffs, adopt cooperation if their cooperative neighbors are successful. Since this behavioral channel does not rely on local network structures to amplify payoffs, it can effectively offset the erosion of benefits caused by highly connected networks.

Our three standard criteria for the evolutionary success of cooperation, $\rho_C > 1/N$, $\rho_D < 1/N$, and $\rho_C > \rho_D$, are equivalent and reduce to the single condition
\begin{equation}\label{eq:cond_pgg}
\Delta \alpha
+3\left[\frac{k+3}{(k+1)^2}b-c\right]>0.
\end{equation}
More precisely, we have \(
\rho_D<\frac{1}{N}<\rho_C
\) if \eqref{eq:cond_pgg} is satisfied and 
\(
\rho_C<\frac{1}{N}<\rho_D
\)
if \eqref{eq:cond_pgg} is reversed. Consequently, scenarios where selection simultaneously favors or disfavors the fixation of both strategies are impossible.

If both types aspire to the same level ($\Delta\alpha = 0$), this condition simplifies to a strict payoff requirement
\begin{equation}\label{eq:bc_threshold}
\frac{b}{c}>\frac{(k+1)^2}{k+3}.
\end{equation}
While this inequality can be met in networks with a small, finite degree, it becomes increasingly difficult to satisfy as $k$ increases. In the limit of large connectivity ($1\ll k\ll N$), condition \eqref{eq:bc_threshold} can never be satisfied, making cooperation impossible to sustain under symmetric aspirations (resulting in $\rho_C < 1/N < \rho_D$).

However, if aspirations are asymmetric, the condition in \eqref{eq:cond_pgg} shows that cooperation can survive even in highly connected graphs. As $k$ becomes large, the invasion condition approaches a simple threshold $\Delta\alpha > 3c$. This highlights a clear separation of mechanisms. The payoff component of the equation dilutes as the degree increases, but the behavioral aspiration component remains intact. If interactions take place in large groups that neutralize the structural benefit of cooperative behavior, only a sufficiently large, positive aspiration gap can compensate for the costs associated with contributing.

To visualize how graph degree and aspiration asymmetry jointly affect cooperator fixation, the left panel in Fig.~\ref{fig:heatmap_selection_advantage} plots the first-order selection advantage \(\rho_C-1/N\) as a function of \(k\) and \(\Delta\alpha\). The neutrality boundary rises with \(k\), illustrating that increasing connectivity progressively raises the aspiration gap required to offset the diluted marginal benefit of cooperation.

These conclusions remain valid when aspiration responsiveness and imitation driven by payoffs have different intensities. Let $\delta_A$ measure the sensitivity of the aspiration gate and let $\delta_I$ measure the strength of payoff-biased imitation. The three criteria $\rho_C>1/N$, $\rho_D<1/N$, and $\rho_C>\rho_D$ remain equivalent and reduce to
\begin{equation}
\Delta\alpha+\frac{k+1+2\lambda(k+4)}{(k+1)^2}b>(1+2\lambda)c,
\end{equation}
where \(\lambda=\delta_I/\delta_A\) for \(\delta_A>0\). Thus, the single-intensity condition \eqref{eq:cond_pgg} is recovered by setting \(\lambda=1\). In the regime \(1\ll k\ll N\), the benefit term vanishes, and the condition reduces to
\begin{equation}
\Delta\alpha>(1+2\lambda)c.
\end{equation}
Hence, separating the two intensities changes the cost threshold but not the dilution mechanism: the benefit of cooperating is still suppressed by increasing degree, whereas the aspiration gap remains independent of degree. The detailed derivation is given in Appendix~\ref{appendixE-part1}.

\begin{figure}[H]
\centering
\includegraphics[width=\textwidth,height=6cm,]{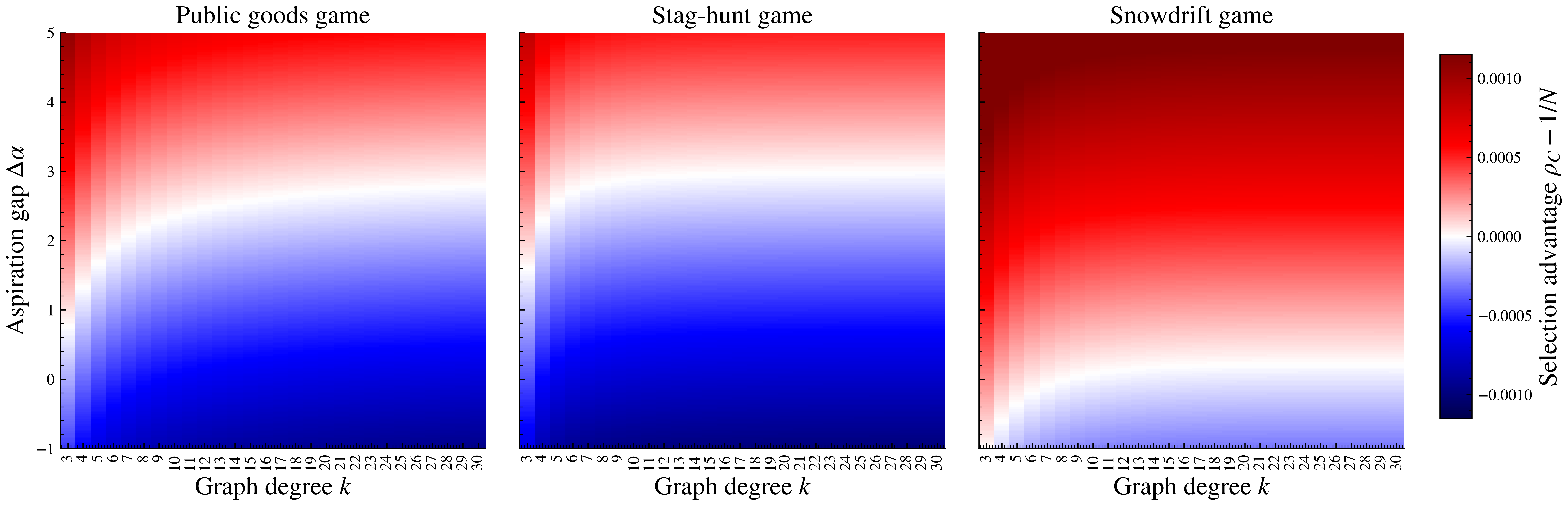}
\caption{
The heatmap represents the first-order selection advantage \(\rho_C-1/N\) 
as a function of graph degree \(k\) and aspiration gap
\(\Delta\alpha=\alpha_D-\alpha_C\), with \(b=2\), \(c=1\), and
\(\delta=10^{-3}\). For $\rho_C$, we have used the first-order approximations in Eqs. 
\eqref{eq:fix_pgg_asp}, \eqref{eq:fix_sh_explicit_final} and \eqref{eq:fix_SD}.
Since \(k\) is discrete, each vertical column
corresponds to one \(k\)-regular graph degree. Red regions indicate
\(\rho_C>1/N\), blue regions indicate \(\rho_C<1/N\), and the white
transition region corresponds to approximate neutrality. Across all three
games, increasing \(\Delta\alpha\) shifts the system toward cooperation,
whereas increasing \(k\) weakens the payoff contribution. The
three panels show that the same dilution mechanism produces different
parameter regions in the public goods, stag-hunt, and snowdrift games.}
\label{fig:heatmap_selection_advantage}
\end{figure}
\section{Stag-hunt game}\label{sec:staghunt}
In this section, we consider the stag-hunt game \cite{skyrms2004stag,pacheco2009evolutionary} on a finite $k$-regular graph. In this game, each cooperator pays a cost $c$ in every interaction. However, the benefit $b$ is generated only under strict full coordination among all group members. Consequently, a cooperator receives $b-c$ if and only if all other players in its interaction group also cooperate; otherwise, it receives only $-c$. Defectors pay no cost and receive $0$ regardless of the group composition. Using the notation from Section~\ref{sec:model}, the payoffs are defined as
\[
a_{n}=
\begin{cases}
-c,& n=0,1,\ldots,k-1,\\
b-c,& n=k,
\end{cases}
\]
and 
$b_{n}=0$ for any $n=0,1,\ldots,k$. 

Under weak selection, applying the diffusion approximation detailed in Appendix~\ref{appendixC-part1} yields 
\begin{subequations}\label{eq:fix_sh_explicit_final}
\begin{align}
\rho_C&\approx\frac1N+\frac{\delta}{4}\Bigg[\Delta\alpha+\;b\,\Phi_C^{\mathrm{SH}}(k)-3c\Bigg],\\
\rho_D&\approx\frac1N-\frac{\delta}{4}\Bigg[\Delta\alpha+\;b\,\Phi_D^{\mathrm{SH}}(k)-3c\Bigg],\\
\frac{\rho_C}{\rho_D}&\approx1+\frac{N\delta}{2}\Bigg[\Delta\alpha+\;b\,\Phi_R^{\mathrm{SH}}(k)-3c\Bigg],
\end{align}
\end{subequations}
where 
\begin{subequations}\label{eq:form_of_phi_SH}
\begin{align}
\Phi_C^{\mathrm{SH}}(k) &=\frac{2(3k+4)((k-1)^{k+2}+3-k^2)}{k\,(k+1)^2\,(k+2)\,(k-2)^2(k-1)^{\,k-1}},\\
\Phi_D^{\mathrm{SH}}(k)&=\frac{2(3k+4)((k-1)^{k+1}(k^2-k-3)+1)}{k\,(k+1)^2\,(k+2)\,(k-2)^2(k-1)^{\,k-1}},\\
\Phi_R^{\mathrm{SH}}(k)&=\frac{(3k+4)((k-1)^{k+1}-1)}{k\,(k+1)^2\,(k-2)(k-1)^{\,k-1}}.
\end{align}
\end{subequations}

As proven in Appendix~\ref{appendixC-part2} and illustrated in Fig.~\ref{fig:staghunt}(a), the coefficients $\Phi_C^{\mathrm{SH}}(k)$, $\Phi_D^{\mathrm{SH}}(k)$, and $\Phi_R^{\mathrm{SH}}(k)$ all decrease as the degree $k$ increases. Their decay formalizes a strong dilution of benefits. Increasing connectivity weakens the payoff advantage that cooperation can generate, while the cost term remains the fixed quantity $3c$. In the limit of large $k$, we have
\[
\Phi_C^{\mathrm{SH}}(k)\sim \frac{6}{k^2},\quad
\Phi_D^{\mathrm{SH}}(k)\sim \frac{6}{k},\quad
\Phi_R^{\mathrm{SH}}(k)\sim \frac{3}{k}.
\]
Hence, the benefit contribution to cooperator invasion decays one power of
$k$ faster than the benefit contributions relevant for defector invasion and
for the comparison $\rho_C>\rho_D$. This asymmetry is the signature of the
coordination barrier in the stag-hunt game. A rare cooperator can benefit only
when its local neighborhood is fully cooperative, and the probability of such
an event becomes extremely small as $k$ grows. By contrast, a rare defector
introduced into an otherwise cooperative population can still exploit the
cooperative environment while simultaneously breaking local unanimity. Thus,
the structure penalizes cooperator invasion much more strongly than it
penalizes defector invasion.

Unlike the linear public goods game, the three standard success criteria are not equivalent because
$\Phi_C^{\mathrm{SH}}(k)\neq\Phi_D^{\mathrm{SH}}(k)\neq\Phi_R^{\mathrm{SH}}(k)$.
They are given by
\begin{subequations}\label{eq:cond_sh_promote_final}
\begin{align}
\Delta\alpha+b\,\Phi_C^{\mathrm{SH}}(k) &> 3c,\\
\Delta\alpha+b\,\Phi_D^{\mathrm{SH}}(k) &> 3c,\\
\Delta\alpha+b\,\Phi_R^{\mathrm{SH}}(k) &> 3c.
\end{align}
\end{subequations}
Since all three $\Phi$-coefficients vanish as $k\to\infty$, these conditions
collapse to the same asymptotic requirement,
\[
\Delta\alpha>3c.
\]
This threshold for a large degree $k$ matches the one found in the public goods game, though it arises from a different mechanism. In the PGG, the benefit is diluted by averaging over many neighbors, whereas in the stag-hunt game, it disappears because local unanimity becomes unattainable for a rare cooperator.

To determine when cooperation is favored, consider first the case $\Delta\alpha>3c$. Then all three inequalities in Eq.~\eqref{eq:cond_sh_promote_final} hold for any $b>0$, resulting in $\rho_D<1/N<\rho_C$. Therefore, the interesting regime is $\Delta\alpha<3c$, where aspirations alone are insufficient. In this case, the critical benefit-to-cost thresholds are
\begin{subequations}\label{eq:bc_star_promote}
\begin{align}
\left(\frac{b}{c}\right)^{*}_{\rho_C>1/N}&=\frac{3-\Delta\alpha/c}{\Phi_C^{\mathrm{SH}}(k)},\\
\left(\frac{b}{c}\right)^{*}_{\rho_D<1/N}&=\frac{3-\Delta\alpha/c}{\Phi_D^{\mathrm{SH}}(k)},\\
\left(\frac{b}{c}\right)^{*}_{\rho_C>\rho_D}&=\frac{3-\Delta\alpha/c}{\Phi_R^{\mathrm{SH}}(k)}.
\end{align}
\end{subequations}
These thresholds are illustrated in Fig.~\ref{fig:staghunt}(b).
A simple comparison of the quantities in Eq.~\eqref{eq:form_of_phi_SH} yields
\[
\Phi_D^{\mathrm{SH}}(k)-\Phi_C^{\mathrm{SH}}(k)
=
\frac{
2(3k+4)\left[(k-1)^{k+1}(k^2-2k-2)+k^2-2\right]
}{
k(k+1)^2(k+2)(k-2)^2(k-1)^{k-1}
}>0,
\]
for $k\geq3$. In addition, we have $\Phi_R^{\mathrm{SH}}(k)=\left(\Phi_D^{\mathrm{SH}}(k)+\Phi_C^{\mathrm{SH}}(k)\right)/2$. This shows that 
\[
\Phi_C^{\mathrm{SH}}(k)<\Phi_R^{\mathrm{SH}}(k)<\Phi_D^{\mathrm{SH}}(k),
\]
so the thresholds satisfy the ordering
\[
\left(\frac{b}{c}\right)^{*}_{\rho_D<1/N}<\left(\frac{b}{c}\right)^{*}_{\rho_C>\rho_D}<\left(\frac{b}{c}\right)^{*}_{\rho_C>1/N}.
\]
This produces a clear sequence of evolutionary regimes as $b/c$ increases (Fig.~\ref{fig:staghunt}(b)). In particular, the threshold for favoring cooperator fixation is always the largest, while that for suppressing defector fixation is always the smallest. In addition, we have
$\Phi_C^{\mathrm{SH}}(k)/\Phi_D^{\mathrm{SH}}(k)\sim 1/k$, showing that for large $k$ the benefit-to-cost
ratio needed to favor a rare cooperator is about $k$ times larger than the ratio needed to suppress a rare defector.

Since the $\Phi$-coefficients decrease to zero, all three thresholds diverge with $k$ according to the following asymptotic behaviors
\begin{align*}
&\left(\frac{b}{c}\right)^{*}_{\rho_C>1/N}\sim \left(3-\Delta\alpha/c\right)\frac{k^2}{6},\\
&\left(\frac{b}{c}\right)^{*}_{\rho_D<1/N}\sim \left(3-\Delta\alpha/c\right)\frac{k}{6},\\
&\left(\frac{b}{c}\right)^{*}_{\rho_C>\rho_D}\sim \left(3-\Delta\alpha/c\right)\frac{k}{3}.
\end{align*}
Hence, when $\Delta\alpha<3c$, making cooperation favored from rarity becomes
quadratically harder as connectivity increases, whereas suppressing rare
defectors requires only a linearly increasing benefit.

The middle panel of Fig.~\ref{fig:heatmap_selection_advantage} complements this result by showing how the parameters interact: as \(k\) increases, the neutrality boundary moves upward, showing that a larger aspiration gap is needed to compensate for the increasingly severe coordination barrier.

It is useful to note that the above conclusions do not require the sensitivity to aspirations and the imitation based on payoffs to have the same intensity. In the case with two different intensities $\delta_A$ and $\delta_I$, the three success conditions become
\begin{equation}\label{eq:sh_two_intensity_condition_main}
\Delta\alpha+b\,\frac{k+2\lambda(k+2)}{3k+4}\,\Phi_S^{\mathrm{SH}}(k)>c(1+2\lambda),
\end{equation}
where $S\in\{C,D,R\}$ refers respectively to the criteria $\rho_C>1/N$, $\rho_D<1/N$, and $\rho_C>\rho_D$. Hence, if $\Delta\alpha/c<1+2\lambda$, the corresponding critical benefit-to-cost ratios are
\begin{equation}\label{eq:sh_two_intensity_threshold_main}
\left(\frac{b}{c}\right)^*_S=\frac{(3k+4)(1+2\lambda-\Delta\alpha/c)}{\bigl(k+2\lambda(k+2)\bigr)\Phi_S^{\mathrm{SH}}(k)},
\qquad
S\in\{C,D,R\}.
\end{equation}
When $\lambda=1$, these expressions simplify and we fall back onto the original baseline thresholds from Eq.~\eqref{eq:bc_star_promote}. Since $\Phi_C^{\mathrm{SH}}(k)\sim 6/k^2$, $\Phi_D^{\mathrm{SH}}(k)\sim 6/k$, and $\Phi_R^{\mathrm{SH}}(k)\sim 3/k$, the thresholds with two intensities scale as
\begin{align}
\left(\frac{b}{c}\right)_{\rho_C>1/N}^{*}
&\sim
\frac{1+2\lambda-\Delta\alpha/c}{2(1+2\lambda)}\,k^2, \label{eq:sh_two_intensity_asymp_C_main}\\
\left(\frac{b}{c}\right)_{\rho_D<1/N}^{*}
&\sim
\frac{1+2\lambda-\Delta\alpha/c}{2(1+2\lambda)}\,k, \label{eq:sh_two_intensity_asymp_D_main}\\
\left(\frac{b}{c}\right)_{\rho_C>\rho_D}^{*}
&\sim
\frac{1+2\lambda-\Delta\alpha/c}{1+2\lambda}\,k. \label{eq:sh_two_intensity_asymp_R_main}
\end{align}
Thus, the model with two intensities preserves the same hierarchy as in the single-intensity case: favoring a rare cooperator remains quadratically harder in $k$, while suppressing rare defectors remains linearly harder. The proof is provided in Appendix~\ref{appendixE-part2}.

 \begin{figure}
\centering
\includegraphics[ width=\textwidth, keepaspectratio]{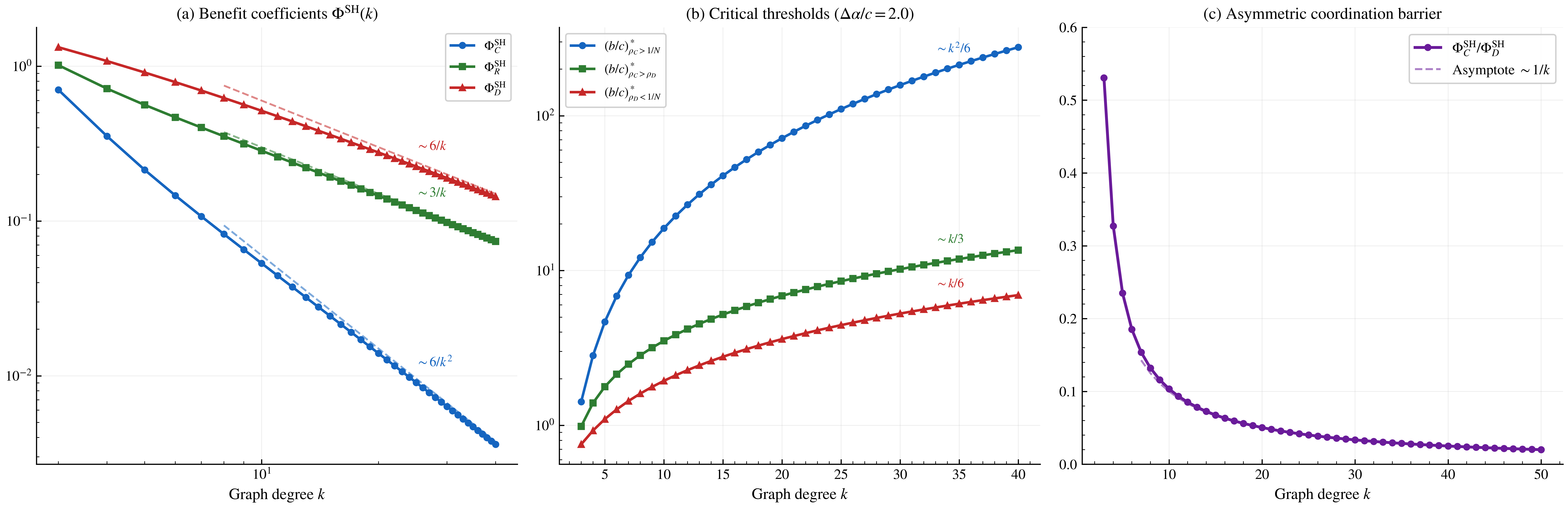}
\caption{Stag-hunt game.
Increasing graph degree strengthens the asymmetric coordination barrier against rare cooperators.
(a)~The benefit coefficients $\Phi_C^{\mathrm{SH}}(k)$, $\Phi_R^{\mathrm{SH}}(k)$, and $\Phi_D^{\mathrm{SH}}(k)$ all decrease with graph degree, but $\Phi_C^{\mathrm{SH}}(k)$, which controls cooperator invasion from rarity, decays much faster than the coefficients associated with defector invasion and the ratio of the fixation probabilities. This reflects the fact that rare cooperators require full local coordination to benefit.
(b) The critical benefit-to-cost thresholds increase with \(k\). The threshold for selection to favor cooperator fixation is the largest, while the threshold for selection to disfavor defector fixation is the smallest. Increasing graph degree makes the emergence of cooperation from rarity much harder than the suppression of rare defectors.
(c)~The decreasing ratio $\Phi_C^{\mathrm{SH}}(k)/\Phi_D^{\mathrm{SH}}(k)$ directly quantifies this asymmetry: as $k$ grows, the payoff leverage of rare cooperators becomes negligible relative to that of rare defectors.}
\label{fig:staghunt}
\end{figure}

\section{Snowdrift game}\label{sec:snowdrift}
In this section, we consider the snowdrift game \cite{zheng2007cooperative,souza2009evolution,santos2011risk} on a $k$-regular graph under weak selection. In this game, a group can clear a ``snowdrift'' at a total cost $c$, generating a benefit $b$ for everyone provided at least one individual cooperates. If the group contains $n+1$ cooperators (including the focal individual), the cost is shared equally among them. Consequently, a cooperator pays $c/(n+1)$ while all group members receive the benefit $b$. A defector incurs no cost and receives the benefit $b$ provided there is at least one cooperator in the interaction group. Thus the payoffs are
\[
a_{n}=b-\dfrac{c}{n+1}\quad\text{for }n=0,1,\ldots,k,
\]
and
\[
b_{n}=\begin{cases}
0,& n=0,\\
b,& n=1,\ldots,k.
\end{cases}
\]

Using the diffusion approximation (see Appendix~\ref{AppendixD}), the first-order expansions of the fixation probabilities are
\begin{subequations}\label{eq:fix_SD}
\begin{align}
\rho_C&\approx\frac1N+\frac{\delta}{4}\Bigg[\Delta\alpha+\;b\,\Phi_C^{\mathrm{SD}}(k)\;-\;c\,\Gamma_{C}^{\mathrm{SD}}(k)\Bigg],\\
\rho_D&\approx\frac1N-\frac{\delta}{4}\Bigg[\Delta\alpha+\;b\,\Phi_D^{\mathrm{SD}}(k)\;-\;c\,\Gamma_{D}^{\mathrm{SD}}(k)\Bigg],\\
\frac{\rho_C}{\rho_D}&\approx1+\frac{N\delta}{2}\Bigg[\Delta\alpha+\;b\,\Phi_R^{\mathrm{SD}}(k)\;-\;c\,\Gamma_{R}^{\mathrm{SD}}(k)\Bigg],
\end{align}
\end{subequations}
where the closed forms for the benefit weights are
\begin{subequations}\label{snowdrift_phi}
\begin{align}
\Phi_C^{\mathrm{SD}}(k)
=&\frac{2(3k+4)((k-1)^{k+1}(k^2-k-3)+1)}{k(k-2)^2(k+1)^{2}(k+2)(k-1)^{k-1}},\\
\Phi_D^{\mathrm{SD}}(k)
=&\frac{2(3k+4)((k-1)^{k+2}+3-k^2)}{k(k-2)^2(k+1)^2(k+2)(k-1)^{k-1}},\\
\Phi_R^{\mathrm{SD}}(k)
=&\frac{(3k+4)((k-1)^{k+1}-1)}{k(k-2)(k+1)^2(k-1)^{k-1}}.
\end{align}
\end{subequations}
The cost weights $\Gamma_{C}^{\mathrm{SD}}(k)$, $\Gamma_{D}^{\mathrm{SD}}(k)$, and $\Gamma_{R}^{\mathrm{SD}}(k)$ are given explicitly in Eqs.~\eqref{eq:I1I2_final} and \eqref{eq:I3_final} in Appendix~\ref{AppendixD}.

Note that $\Phi_C^{\mathrm{SD}}(k)$, $\Phi_R^{\mathrm{SD}}(k)$, and $\Phi_D^{\mathrm{SD}}(k)$ are all decreasing functions of $k$. See Remark \ref{snowdrift_coefficient} in Appendix \ref{appendixD-part2} for the analytical proof, and Fig.~\ref{fig:snowdrift}(a) for numerical illustrations. Moreover, the closed forms in Eq.~\eqref{snowdrift_phi} imply the asymptotics
\begin{subequations}\label{eq:snowdrift_approx1}
\begin{align}
&\Phi_C^{\mathrm{SD}}(k)\sim 2\Phi_R^{\mathrm{SD}}(k)\sim \frac{6}{k},
\\
&\Phi_D^{\mathrm{SD}}(k)\sim \frac{6}{k^2},
\end{align}
\end{subequations}
for $k$ large.
Thus, the benefit contribution becomes negligible in all three invasion
statistics as connectivity increases, and it decays fastest from
$\rho_D$. This monotone decay formalizes the dilution of the benefit. In the snowdrift game, the benefit is produced whenever at least one cooperator is present. As the degree increases, the focal player becomes less important for determining whether this condition is satisfied, so the benefit contribution to selection is diluted. This dilution is strongest for a rare defector in a mostly cooperative population, because most groups already contain other cooperators; by contrast, a rare cooperator in a mostly defecting population can still create benefit in groups that would otherwise receive none.

Fig.~\ref{fig:snowdrift}(b) suggests that the cost weights $\Gamma_{C}^{\mathrm{SD}}(k)$, $\Gamma_{D}^{\mathrm{SD}}(k)$ and $\Gamma_{R}^{\mathrm{SD}}(k)$ also decrease with $k$. Furthermore, in Appendix~\ref{appendixD-part3}, we have shown that
\begin{subequations}\label{eq:snowdrift_approx2}
\begin{align}
&\Gamma_C^{\mathrm{SD}}(k)\sim 2\Gamma_R^{\mathrm{SD}}(k)\sim\frac{6\ln k}{k},\\
&\Gamma_D^{\mathrm{SD}}(k)\sim \frac{6}{k},
\end{align}
\end{subequations}
for $k$ large. Hence, the cost contribution in \eqref{eq:fix_SD} vanishes in the limit of large degree, but more slowly for $\rho_C$ and $\rho_C/\rho_D$ because of the logarithmic factor. The mechanism is the cost-sharing structure of the snowdrift game: when several cooperators are present, the cost is split among them, so the per-capita burden decreases with the number of contributors. Averaging terms of the form $1/j$ over possible numbers of cooperating group members generates harmonic sums, which explains the appearance of $\ln k$. The $\Gamma$-coefficients differ because they correspond to different invasion contexts. In particular, $\Gamma_C^{\mathrm{SD}}(k)$ measures how strongly shared costs constrain a cooperator lineage across frequencies, whereas $\Gamma_D^{\mathrm{SD}}(k)$ measures the advantage enjoyed by a defector that avoids those costs in a mostly cooperative environment. Consistent with this interpretation, $\Gamma_D^{\mathrm{SD}}(k)$ is smaller and decays faster than $\Gamma_C^{\mathrm{SD}}(k)$ and $\Gamma_R^{\mathrm{SD}}(k)$.

For \(\Delta\alpha\neq 0\), the limit \(1\ll k\ll N\) eliminates the terms dependent on payoffs in \eqref{eq:fix_SD}. Hence, we have
\[
\begin{aligned}
\rho_C &\approx \frac1N+\delta\,\frac{\Delta\alpha}{4},\\
\rho_D &\approx \frac1N-\delta\,\frac{\Delta\alpha}{4}.
\end{aligned}
\]
Thus, in highly connected graphs, fixation is governed primarily by the aspiration gap $\Delta\alpha$. This contrasts sharply with the stag-hunt game, where a nonzero cost threshold persists as $k\to\infty$ because costs are incurred per interaction and are not shared. In the snowdrift game, sharing the cost dilutes the effective burden and entirely removes the impact of payoffs when the network degree becomes large.

Next, consider the symmetric case \(\Delta\alpha=0\). The critical benefit-to-cost thresholds are
\begin{subequations}\label{eq:snowdrift_approx3}
\begin{align}
\left(\frac{b}{c}\right)^{*}_{\rho_C>1/N}
&=
\frac{\Gamma_C^{\mathrm{SD}}(k)}{\Phi_C^{\mathrm{SD}}(k)},\\
\left(\frac{b}{c}\right)^{*}_{\rho_C>\rho_D}
&=
\frac{\Gamma_R^{\mathrm{SD}}(k)}{\Phi_R^{\mathrm{SD}}(k)},\\
\left(\frac{b}{c}\right)^{*}_{\rho_D<1/N}
&=
\frac{\Gamma_D^{\mathrm{SD}}(k)}{\Phi_D^{\mathrm{SD}}(k)}.
\end{align}
\end{subequations}
Numerical evaluation of the exact coefficient formulas in Appendix~\ref{AppendixD} shows that, over the range of \(k\) displayed in Fig.~\ref{fig:snowdrift}(c), these thresholds satisfy
\[
\left(\frac{b}{c}\right)^{*}_{\rho_C>1/N}
<
\left(\frac{b}{c}\right)^{*}_{\rho_C>\rho_D}
<
\left(\frac{b}{c}\right)^{*}_{\rho_D<1/N}.
\]
This ordering implies a clear progression of evolutionary regimes as \(b/c\) increases over the plotted range. For small \(b/c\), cooperation is disfavored from rarity, and rare defectors are favored. Once
\(\left(b/c\right)^{*}_{\rho_C>1/N}\) is crossed, cooperation becomes favored against neutrality. Crossing
\(\left(b/c\right)^{*}_{\rho_C>\rho_D}\) then makes cooperator fixation more likely than defector fixation. Only beyond the largest threshold,
\(\left(b/c\right)^{*}_{\rho_D<1/N}\), does rare defection itself become disfavored.

By applying the asymptotic approximations in Eqs.~\eqref{eq:snowdrift_approx1} and~\eqref{eq:snowdrift_approx2}, the asymptotic conclusion is
\[
\begin{aligned}
&\left(\frac{b}{c}\right)^{*}_{\rho_C>1/N}\sim \left(\frac{b}{c}\right)^{*}_{\rho_C>\rho_D}\sim\ln k,\\
&\left(\frac{b}{c}\right)^{*}_{\rho_D<1/N}\sim k.
\end{aligned}
\]
Therefore, in the large-degree, locally tree-like regime, the two thresholds involving cooperator success scale logarithmically with \(k\), whereas the threshold required for selection to disfavor defector fixation scales linearly. Thus, the asymptotic separation follows from the large-degree analysis, while the finite-\(k\) ordering shown in Fig.~\ref{fig:snowdrift}(c) should be interpreted as numerical evidence obtained from the exact formulas.

The right panel of Fig.~\ref{fig:heatmap_selection_advantage} shows that for the snowdrift game, the cooperation favored region is larger than in the public goods and stag-hunt cases. This reflects the fact that both benefit and cost contributions dilute as $k$ increases.

The same large-degree conclusion holds for the two-intensity extension $\delta_A$ and $\delta_I$. In this case, the fixation probabilities take the form
\begin{subequations}\label{eq:sd_two_intensity_fix_main}
\begin{align}
\rho_C&\approx\frac1N+\frac{\delta_A}{4}\left[\Delta\alpha+b\Phi_{C,A}^{\mathrm{SD}}(k)-c\Gamma_{C,A}^{\mathrm{SD}}(k)\right]+\frac{\delta_I}{4}\left[b\Phi_{C,I}^{\mathrm{SD}}(k)-c\Gamma_{C,I}^{\mathrm{SD}}(k)\right],\\
\rho_D&\approx\frac1N-\frac{\delta_A}{4}\left[\Delta\alpha+b\Phi_{D,A}^{\mathrm{SD}}(k)-c\Gamma_{D,A}^{\mathrm{SD}}(k)\right]-\frac{\delta_I}{4}\left[b\Phi_{D,I}^{\mathrm{SD}}(k)-c\Gamma_{D,I}^{\mathrm{SD}}(k)\right].
\end{align}
\end{subequations}
Here the benefit coefficients split in the same way as in the stag-hunt game,
\begin{subequations}\label{eq:sd_two_intensity_phi_split_main}
\begin{align}
&\Phi_{S,A}^{\mathrm{SD}}(k)=\frac{k}{3k+4}\Phi_S^{\mathrm{SD}}(k),\\
&\Phi_{S,I}^{\mathrm{SD}}(k)=\frac{2(k+2)}{3k+4}\Phi_S^{\mathrm{SD}}(k),
\end{align}
\end{subequations}
for $S\in\{C,D,R\}$, while the cost coefficients separate into an aspiration-gated part and an imitation-bias part as
\[
\Gamma_S^{\mathrm{SD}}(k)=\Gamma_{S,A}^{\mathrm{SD}}(k)+\Gamma_{S,I}^{\mathrm{SD}}(k).
\]

The important point is that all benefit and cost coefficients still vanish in the regime $1\ll k\ll N$, giving
\begin{subequations}\label{eq:sd_two_intensity_dense_main}
\begin{align}
&\rho_C\approx\frac1N+\frac{\delta_A\Delta\alpha}{4},\\
&\rho_D\approx\frac1N-\frac{\delta_A\Delta\alpha}{4}.
\end{align}
\end{subequations}
Therefore, the conditions for $\rho_C > 1/N$, $\rho_D < 1/N$, and $\rho_C > \rho_D$ are equivalent and reduce to
\begin{equation}\label{eq:sd_two_intensity_dense_condition_main}
\Delta\alpha>0.
\end{equation}
Although separating aspiration responsiveness from imitation bias changes the weights of the payoff terms for finite \(k\), it does not alter the main snowdrift conclusion. On regular graphs of large degree, payoff effects dilute away at first order, leaving aspiration asymmetry as the leading contribution to the direction of selection. The full derivation is given in Appendix~\ref{appendixE-part3}.

\begin{figure}[H]
\centering
\includegraphics[height=6cm, width=\textwidth, keepaspectratio]{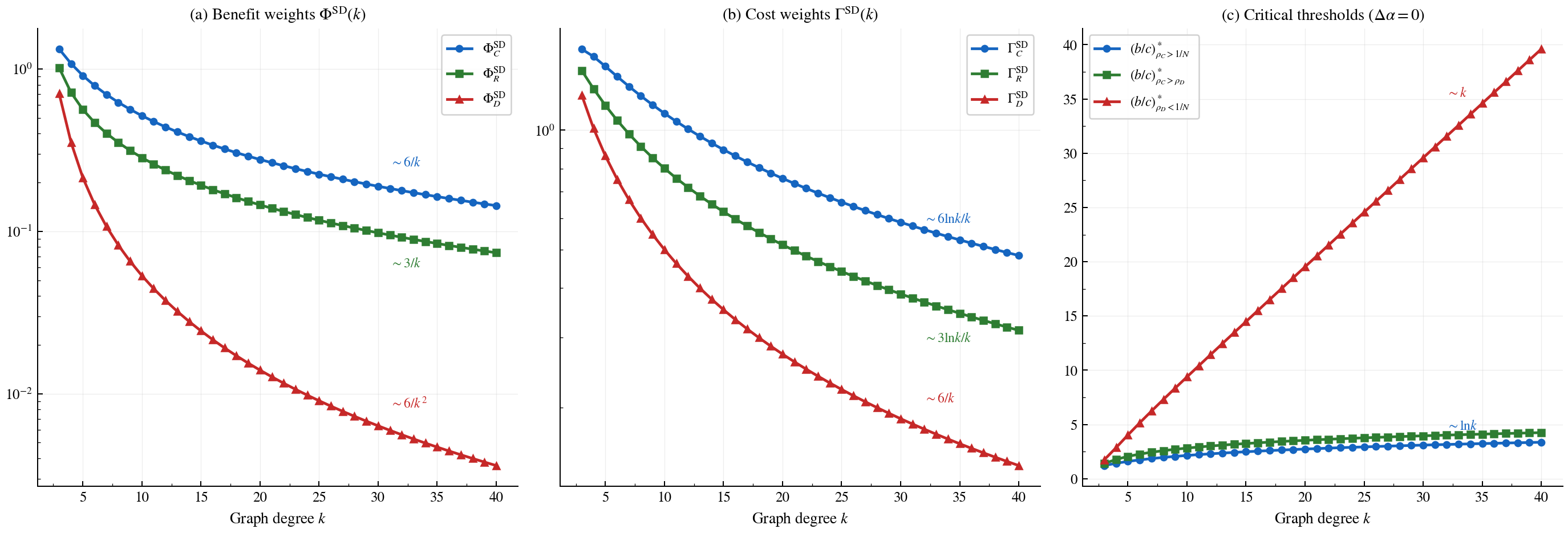}
\caption{Snowdrift game.
The figure shows that increasing graph degree dilutes both benefit weights $\Phi^{\mathrm{SD}}$ and cost weights $\Gamma^{\mathrm{SD}}$, but with different consequences for the three success criteria.
(a)~The benefit weights $\Phi_C^{\mathrm{SD}}(k)$, $\Phi_R^{\mathrm{SD}}(k)$, and $\Phi_D^{\mathrm{SD}}(k)$ decrease with $k$, because the focal individual becomes less decisive for producing the group benefit in larger neighborhoods. The coefficient $\Phi_D^{\mathrm{SD}}(k)$ decays fastest as a mostly cooperative population typically already contains other individuals producing benefits.
(b)~The cost weights \(\Gamma_C^{\mathrm{SD}}(k)\), \(\Gamma_R^{\mathrm{SD}}(k)\), and \(\Gamma_D^{\mathrm{SD}}(k)\) suggest the same decreasing trend in the plotted range, reflecting cost sharing among cooperators. The slower decay of $\Gamma_C^{\mathrm{SD}}(k)$ and $\Gamma_R^{\mathrm{SD}}(k)$ comes from the structure generated by averaging over possible numbers of contributors.
(c)~Under symmetric aspirations ($\Delta\alpha=0$), the thresholds $(b/c)^*$ calculated from the exact coefficient formulas diverge over the plotted range: $(b/c)^*_{\rho_C>1/N}$ and $(b/c)^*_{\rho_C>\rho_D}$ grow slowly, whereas $(b/c)^*_{\rho_D<1/N}$ increases much faster.  In the large-degree limit, the corresponding asymptotics are $\ln k$, $\ln k$, and $k$, respectively. Thus, favoring the fixation of cooperation or favoring cooperation over defection remains possible at moderate benefit-to-cost ratios, but disfavoring the fixation of defection becomes increasingly difficult as $k$ increases.}
\label{fig:snowdrift}
\end{figure}

\section{\label{sec:simulations}Monte Carlo Simulations}
To examine the first-order predictions under weak selection, we performed Monte Carlo simulations on connected simple random \(k\)-regular graphs across a broad range of connectivity values. We distinguish two purposes of the numerical analysis. Simulations satisfying \(N\delta\ll1\) and \(k\ll N\) are used to assess the quantitative accuracy of the first-order approximation, whereas simulations with larger \(N\delta\) are used only to examine whether the predicted qualitative dependence on graph degree remains visible outside the formal approximation regime.

Unless stated otherwise, all simulations used the parameters \(b=2\), \(c=1\), \(\alpha_C=0\), and \(\delta=10^{-3}\), with \(10^6\) stochastic update runs per parameter combination, organized in graph batches of \(10^3\) runs per random graph. At the start of each batch, a new simple random \(k\)-regular graph on \(N\) vertices was generated using the configuration model. For each node, \(k\) half-edges were created, randomly shuffled, and paired sequentially to form edges. Configurations producing self-loops, multiple edges, or disconnected graphs were discarded, and the procedure was repeated until a connected simple \(k\)-regular graph was obtained. This connectedness requirement ensures that a single mutant can reach the entire population. Without this condition, a mutant initially placed in one connected component could not spread to other components, which would artificially reduce global fixation probabilities.

\subsection{The fixation of cooperation}
Within each run, a single cooperator was placed uniformly at random on one vertex of an otherwise all-defector population, and the dynamics then evolved under the update rule of Section~\ref{sec:model} until absorption. The cooperator fixation probability \(\rho_C\) was estimated as the fraction of runs that ended in the all-cooperator state. We ran the simulations for \(k\in\{3,4,5,6\}\), \(N\in\{100,500\}\), and \(\alpha_D\in\{1.25,1.50,\ldots,3.00\}\) for the three games. The simulated results and the first-order approximations are plotted in Fig.~\ref{fig:sim_k6}.

The simulations for \(\rho_C\) test the first-order approximations in the regime where they are expected to be most accurate, namely \(N\delta\ll 1\) together with \(k\ll N\). The first condition ensures that selection remains perturbative over the fixation process, while the second condition is required by the pair approximation and the locally tree-like assumption used to reduce the process to a one-dimensional diffusion. For \(N=100\) and \(\delta=10^{-3}\), we have \(N\delta=0.1\), and for all simulated degrees \(k=3,4,5,6\), the condition \(k\ll N\) is also well satisfied. In this regime, the Monte Carlo estimates are broadly consistent with the first-order curves for the public goods, stag-hunt, and snowdrift games. The simulations also reproduce the predicted increase of \(\rho_C\) with \(\alpha_D\): larger defector aspirations make defectors less satisfied and therefore more likely to revise, increasing the opportunities for cooperative mutants to spread.

For \(N=500\) and \(\delta=10^{-3}\), the graph condition \(k\ll N\) remains well satisfied, but the weak-selection condition is less stringent because \(N\delta=0.5\). Therefore, these simulations should not be interpreted as a strict quantitative test of the first-order expansion. They are better viewed as a check of the qualitative trends beyond the formal regime \(N\delta\ll1\). For \(k=3\) and \(k=4\), visible deviations appear in all three games. These deviations become considerably smaller for \(k=5\) and \(k=6\). This pattern is consistent with the role of graph degree in the approximation. At low degree, local neighborhoods are small, local correlations are stronger, and the payoff effect of a single mutant is less diluted. Higher-order terms can therefore have a more visible effect. As \(k\) increases, payoff effects are averaged over larger neighborhoods and the effective selection contribution is weakened, so the first-order curves tend to provide a better approximation. Despite the quantitative deviations at low degree, the simulations retain the main qualitative pattern: increasing \(\alpha_D\) increases \(\rho_C\), and the role of graph degree remains consistent with the dilution mechanism.

\begin{figure}[H]
  \centering
  \includegraphics[width=\textwidth, keepaspectratio]{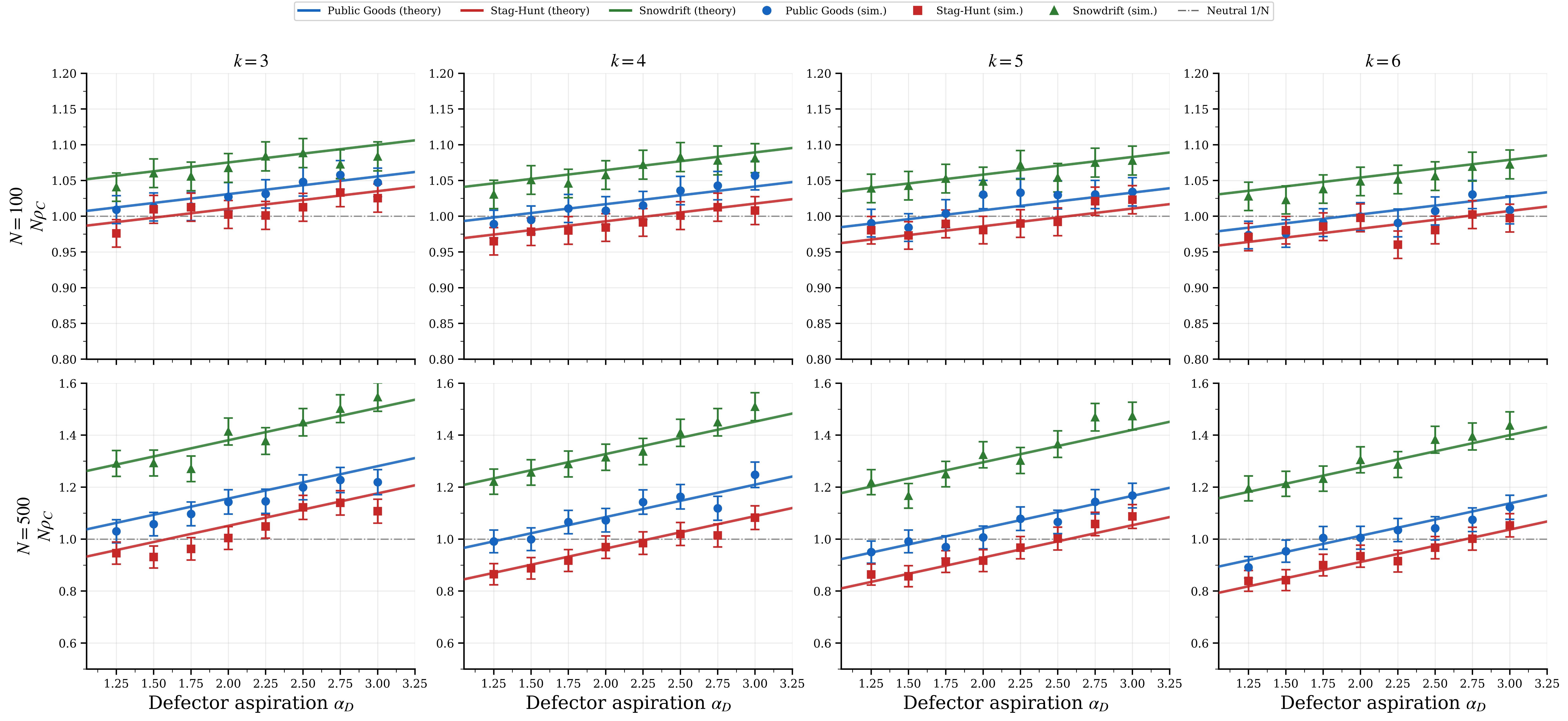}
\caption{
Monte Carlo estimates of the cooperator fixation probability \(\rho_C\).
A single cooperator is placed uniformly at random in an otherwise all-defector population on a connected simple random \(k\)-regular graph.
The simulations are shown for the public goods, stag-hunt, and snowdrift games, with \(k=3,4,5,6\), \(N\in\{100,500\}\), and \(\alpha_D\in\{1.25,1.50,\ldots,3.00\}\).
The remaining parameters are \(b=2\), \(c=1\), \(\alpha_C=0\), and \(\delta=10^{-3}\).
Points show Monte Carlo estimates with 95\% confidence intervals computed over graph batches.
Solid lines show the first-order predictions under weak selection from Eqs.~\eqref{eq:fix_pgg_asp}, \eqref{eq:fix_sh_explicit_final}, and \eqref{eq:fix_SD}.
The dashed horizontal line marks the neutral fixation probability \(1/N\).
Agreement is better for \(N=100\), where \(N\delta=0.1\), while clear deviations appear for \(N=500\), where \(N\delta=0.5\).
}
  \label{fig:sim_k6}
\end{figure}

\subsection{The fixation of defection}
To complement the cooperator invasion simulations, we also simulated the defector fixation probability \(\rho_D\). In each Monte Carlo run, one vertex was chosen uniformly at random and initialized as a defector, while all remaining vertices were initialized as cooperators. The stochastic update process was then iterated until absorption, and \(\rho_D\) was estimated as the fraction of runs that ended in the all-defector state. The results for \(N=100\), all three games, and \(k=3,4,5,6\) are shown in Fig.~\ref{fig:sim_rhoD}. For \(N=100\) and \(\delta=10^{-3}\), these simulations provide a quantitative comparison with the first-order approximation within the small-\(N\delta\) regime, since \(N\delta=0.1\) and \(k\ll N\).

The Monte Carlo estimates are broadly consistent with the first-order predictions. In all three games, \(\rho_D\) decreases as \(\alpha_D\) increases. This is expected because larger defector aspirations make defectors less satisfied and therefore more likely to revise. The simulations also reproduce the predicted role of graph degree. For the public goods and stag-hunt games, increasing \(k\) weakens the benefit term while the cost term remains present. In the snowdrift game, increasing \(k\) weakens both the benefit and cost terms. Thus, the defector invasion simulations show the same dilution mechanism as the cooperator invasion simulations, but from the opposite resident state.

Since the ratio \(\rho_C/\rho_D\) is determined by the two invasion probabilities, the agreement observed for \(\rho_C\) and \(\rho_D\) provides indirect numerical support for the first-order prediction of the ratio within the tested weak-selection regime.

\begin{figure}[H]
  \centering
  \includegraphics[width=\textwidth, height=10cm]{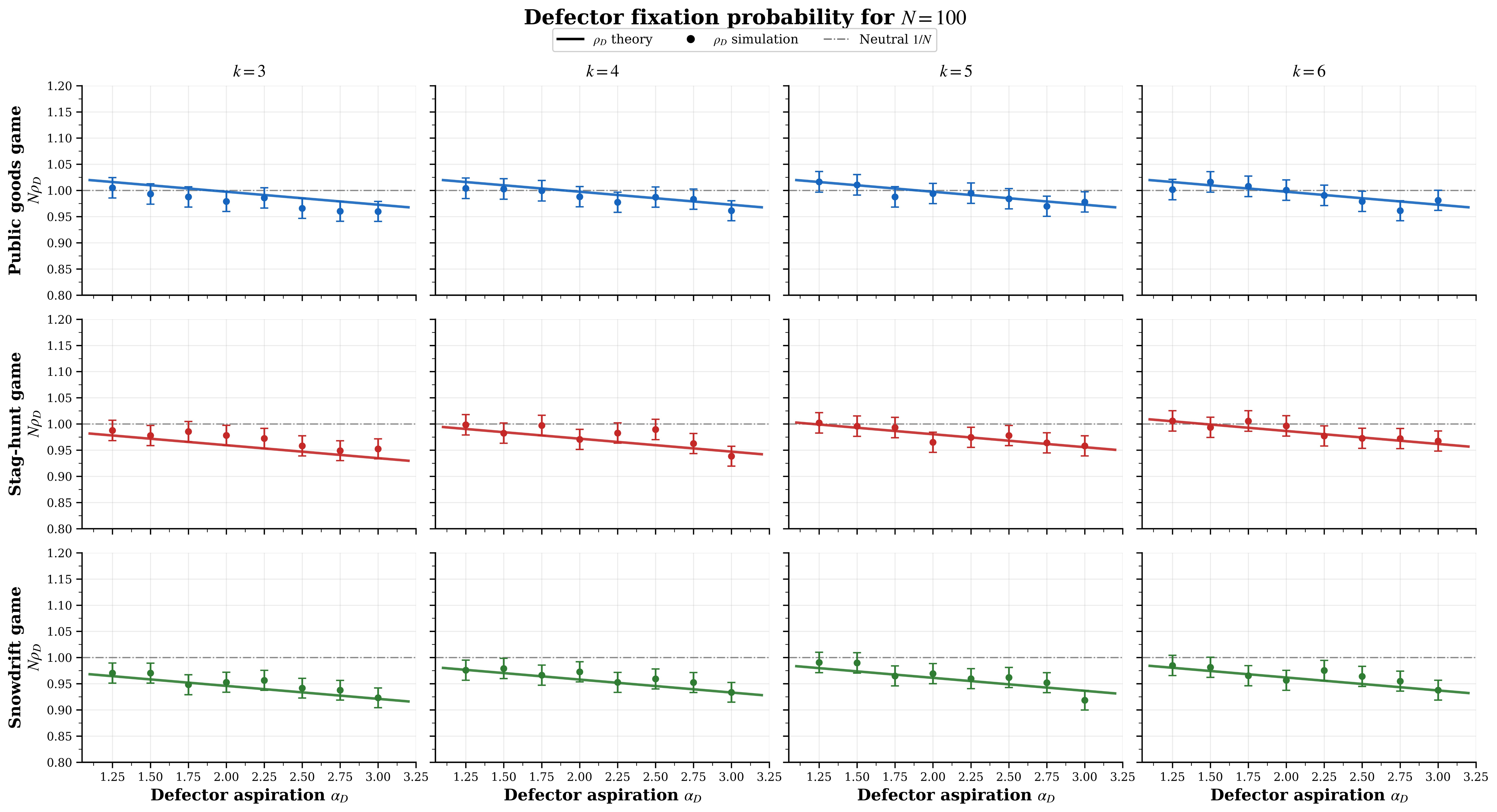}
\caption{
Monte Carlo estimates of the defector fixation probability \(\rho_D\) for \(N=100\).
Rows correspond to the public goods, stag-hunt, and snowdrift games, and columns correspond to \(k=3,4,5,6\).
A single defector is placed uniformly at random in an otherwise all-cooperator population on a connected simple random \(k\)-regular graph.
The remaining parameters are \(b=2\), \(c=1\), \(\alpha_C=0\), \(\delta=10^{-3}\), and \(\alpha_D\in\{1.25,1.50,\ldots,3.00\}\).
Points show Monte Carlo estimates with 95\% confidence intervals computed over graph batches.
Solid lines show the first-order predictions under weak selection from Eqs.~\eqref{eq:fix_pgg_asp}, \eqref{eq:fix_sh_explicit_final}, and~\eqref{eq:fix_SD}.
The dashed horizontal line gives the neutral fixation probability \(1/N\). The simulated values are in agreement with the first-order predictions for \(N=100\), where \(N\delta=0.1\).
}
\label{fig:sim_rhoD}
\end{figure}


\subsection{Sensitivity to selection intensity}

To examine the sensitivity of the analytical approximation to the strength of selection and the degree $k$, we also performed simulations over a small range of selection intensities. The population size and aspiration parameters were fixed at $N=100$, $\alpha_C=0$, and $\alpha_D=2.5$, while the graph degree was varied over $k\in\{3,4,5,6,7,8,9,10,12,15,20\}$ and the selection intensity was set to $\delta\in\{10^{-3},10^{-2},10^{-1}\}$. For each combination of $k$ and $\delta$, we estimated $\rho_C$ for the public goods, stag-hunt, and snowdrift games from Monte Carlo simulations starting with one cooperator in an otherwise all-defector population.  

The purpose of this test is to distinguish the quantitative
accuracy of the first-order expansion from the qualitative robustness
of the predicted degree effect. When \(\delta=10^{-3}\), we have
\(N\delta=0.1\), so these simulations provide the closest quantitative
comparison with the first-order approximation, particularly at the
lower degrees for which \(k/N\) is also small. For
\(\delta=10^{-2}\) and \(\delta=10^{-1}\), however,
\(N\delta=1\) and \(N\delta=10\), respectively. These cases lie outside
the formal first-order regime and are therefore included only as
qualitative trend checks, not as quantitative validations of the
analytical approximation. Across this wider parameter range, the
simulations continue to reproduce the central qualitative prediction
that increasing \(k\) weakens payoff-driven selection in a
game-dependent manner. Thus, the simulations are broadly consistent
with the analytical predictions where \(N\delta\) and \(k/N\) are
small, while the persistence of the degree effect outside this range
provides qualitative evidence only.

\begin{figure}[H]
  \centering
 \includegraphics[height=6cm,width=\textwidth, keepaspectratio]{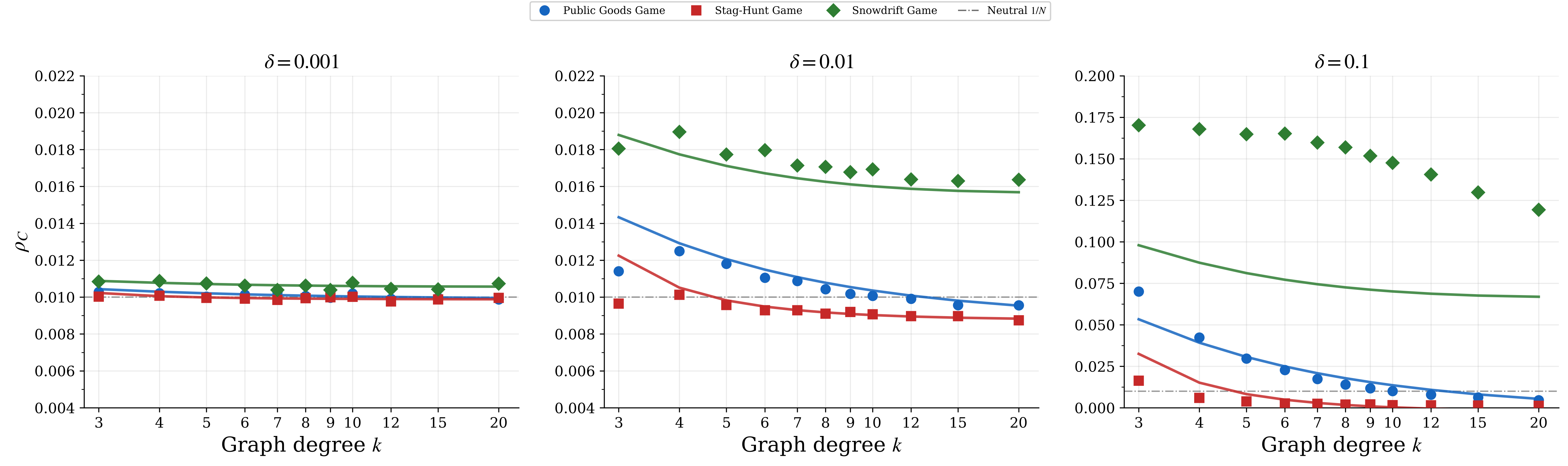}
\caption{
Sensitivity of cooperator fixation to the selection intensity \(\delta\).
The cooperator fixation probability \(\rho_C\) is estimated on connected simple random \(k\)-regular graphs for
\(k\in\{3,4,5,6,7,8,9,10,12,15,20\}\) and
\(\delta\in\{10^{-3},10^{-2},10^{-1}\}\).
Each simulation starts from one cooperator in an otherwise all-defector population, with
\(N=100\), \(b=2\), \(c=1\), \(\alpha_C=0\), and \(\alpha_D=2.5\).
Monte Carlo estimates (dots) are compared with first-order analytical predictions (solid lines).
For \(\delta=10^{-3}\), where \(N\delta=0.1\), the approximation is expected to be most accurate.
For larger \(\delta\), terms beyond first order become visible, but the decrease with graph degree remains apparent.
The cases \(\delta=10^{-2}\) and \(\delta=10^{-1}\) are included as qualitative trend checks rather than as quantitative validations of the first-order approximation.
}
  \label{fig:k_sweep_n100}
\end{figure}

\subsection{\label{subsec:sim_scope}Scope and limitations of the numerical comparisons}

Taken together, the simulations in this section provide a broader
numerical assessment. They cover graph degrees from \(k=3\) to \(k=20\),
the cooperator invasion setting for \(N\in\{100,500\}\), the defector
invasion setting for \(N=100\), all three games, and three values of
the selection intensity in the \(N=100\) sensitivity analysis.

Among these simulations, those with \(N=100\) and
\(\delta=10^{-3}\), for which \(N\delta=0.1\), provide the principal
quantitative comparisons with the first-order formulas, particularly
at the lower degrees for which \(k/N\) is also small. The simulations
with \(N=500\) and \(\delta=10^{-3}\), for which \(N\delta=0.5\), and
the sensitivity runs with \(N=100\) and
\(\delta\in\{10^{-2},10^{-1}\}\), for which
\(N\delta\in\{1,10\}\), lie outside the formal first-order scaling
\(N\delta\ll1\) and are therefore interpreted only as qualitative
trend checks. Similarly, the simulations at the upper end of the
degree range are used primarily to examine the predicted dependence
on \(k\), since \(k/N\) is then no longer very small.

All simulations are performed on connected simple random
\(k\)-regular graphs. The locally tree-like approximation is expected
to be most reliable when \(k/N\) is small and short cycles have a
negligible effect. Clustered networks, heterogeneous-degree networks,
and small-world networks are not considered here. Consequently, the
numerical results should not be interpreted as establishing
quantitative accuracy beyond the tested weak-selection and
sparse-graph settings.
\section{\label{sec:discussion}Discussion}

We studied multiplayer evolutionary games on regular graphs of degree $k\geq3$ under a hybrid aspiration and imitation rule. In this rule, dissatisfaction relative to an aspiration level affects whether an individual revises, while imitation based on payoffs affects which neighbor is copied. Under weak selection and the pair approximation on locally tree-like graphs, the first-order expansion of the fixation probabilities can be written as the sum of two terms. The first term depends on payoffs and varies with the graph degree. The second term is proportional to the aspiration gap \(\Delta\alpha=\alpha_D-\alpha_C\) and does not depend on \(k\).

This decomposition motivates what we call the dilution principle. In the present approximation, increasing the graph degree changes how local payoff events enter selection. In the public goods game, the relevant payoff effect is the marginal contribution of one cooperator to a larger group. In the stag-hunt game, it is the probability of full local coordination. In the snowdrift game, it is the production of a benefit by at least one cooperator together with the sharing of the cost among contributors. These mechanisms differ across games. However, in all three cases, the payoff term in the first-order drift decreases as \(k\) increases, whereas the aspiration term remains unchanged.

The comparison in Fig.~\ref{fig:cross_game} illustrates this decomposition. Under symmetric aspirations, \(\Delta\alpha=0\), the three games have different shapes and orderings because different payoff events enter the drift. In the snowdrift game, cooperation can be favored at low degree because one cooperator may generate the shared benefit. In the stag-hunt game, rare cooperators are disadvantaged because the benefit requires full local coordination. The public goods game gives an intermediate case, where the marginal contribution of a cooperator is diluted as the interaction group grows. Thus, within the present approximation, graph degree filters payoff effects in a way that depends on the game.

Aspiration asymmetry enters in a different way. When \(\Delta\alpha>0\), defectors have a higher aspiration level and are therefore more likely to revise. This creates a bias that can favor cooperation. In Fig.~\ref{fig:cross_game}, this appears as an upward shift of the three curves. Since the aspiration term in the first-order expansion does not depend on \(k\), this shift is not diluted by graph degree. Therefore, aspiration asymmetry can partly compensate for the weakening of payoff effects, especially at low and intermediate degrees.

The three games also differ in their behavior at large degree. In the public goods and stag-hunt games, the benefit terms vanish as \(k\) increases, but the cost term remains in the limiting condition. Cooperation therefore requires a sufficiently positive aspiration gap to offset this remaining cost. In the snowdrift game, both the benefit and cost terms vanish because the cost is shared among cooperators. In this case, the large-degree condition is determined by the sign of \(\Delta\alpha\). This contrast shows that the dilution principle does not take the same form in every game. Its expression depends on the payoff rule and on how benefits and costs enter local group interactions.

The hybrid update rule also places the model between two familiar classes of evolutionary dynamics. Compared with imitation models driven only by payoffs on graphs \cite{OhtsukiHauertLiebermanNowak2006SimpleRule,Chen2013SharpBenefitCost,AllenNowak2014GamesOnGraphs,GokhaleTraulsen2010PNAS,PenaWuArranzTraulsen2016,LiBroomDuWang2016PRE,WangPercSzolnoki2024NatCom}, the present dynamics include an aspiration term that remains visible when payoff effects are weakened by increasing degree. Compared with switching models based on aspiration \cite{KMRV1998,PPS1999,DWAW2014,DuWuWang2015SciRep}, the present dynamics retain an imitation step that depends on payoffs and is shaped by the graph. Thus, the first-order expansion shows how aspiration effects and payoff effects filtered by the network enter together.

The dilution principle is also not tied to the assumption that the two steps of the update rule share the same intensity. In the main model, the single weak-selection parameter \(\delta\) controls both aspiration responsiveness in the revision step and payoff bias in the imitation step.  This is the case of symmetric intensity. If the two channels are separated, with \(\delta_A\) controlling the aspiration gate and \(\delta_I\) controlling payoff-biased imitation, the relative weights of the two mechanisms, and hence the finite-degree success thresholds, depend on the ratio \(\delta_I/\delta_A\). This reweighting does not alter their different responses to graph degree: payoff-dependent contributions retain their game-specific dependence on \(k\), whereas the direct aspiration-gap contribution remains independent of \(k\). Thus, the persistence of the aspiration contribution is not a consequence of imposing the same intensity on the two behavioral steps.

\begin{figure}[H]
\centering
\includegraphics[height=6cm, width=\textwidth, keepaspectratio]{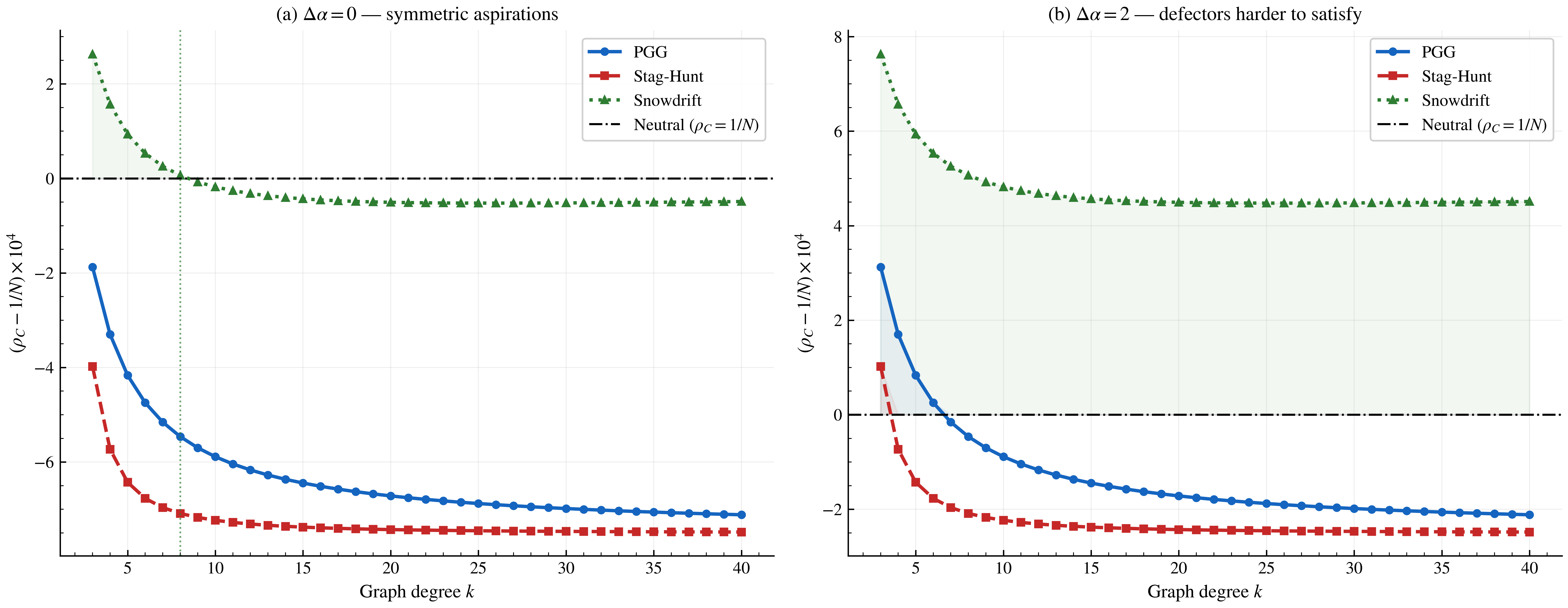}
\caption{Comparison across games of the selection advantage
\(\rho_C - 1/N\) as a function of graph degree \(k\)
(parameters: \(b=2\), \(c=1\), \(N=100\), \(\delta=10^{-3}\)).
The figure illustrates the dilution principle across three payoff mechanisms:
marginal contribution in the public goods game, full coordination in the stag-hunt game, and benefit production with cost sharing in the snowdrift game.
(a)~Under symmetric aspirations, \(\Delta\alpha=0\), the ordering and shape of the curves are determined by the payoff terms of each game. Snowdrift cooperation is favored at low degree but loses its advantage as \(k\) increases. The public goods and stag-hunt games remain below neutrality over the displayed range.
(b)~Under asymmetric aspirations, \(\Delta\alpha=2\), all curves shift upward by the same aspiration term, which does not depend on degree. The payoff structure still determines the relative ordering, but aspiration asymmetry can move some cases from disfavoring to favoring cooperator fixation. Shaded regions indicate where selection favors cooperator fixation, \(\rho_C>1/N\).}
\label{fig:cross_game}
\end{figure}

The framework also allows aspirations to vary within each strategic class. Suppose that individuals using strategy \(S\in\{C,D\}\) draw their aspiration level from a distribution \(f_S(\alpha)\) with mean \(m_S\) and variance \(\sigma_S^2\). For a realized payoff \(\pi\), the revision probability is
\[
F_S(\pi)=\int\frac{1}{1+e^{\delta(\pi-\alpha)}}f_S(\alpha)\,d\alpha .
\]
Under weak selection,
\[
F_S(\pi)=\frac12-\frac{\delta}{4}(\pi-m_S)+O(\delta^2).
\]
Thus, to first order, aspiration heterogeneity enters only through the mean aspiration level of each strategy. The previous results carry over after replacing \(\Delta\alpha\) by \(m_D-m_C\). Variances and higher moments enter only at higher order in \(\delta\). In this sense, the first-order conclusions are unchanged by variation in aspiration levels within each strategy.

Our analytical framework relies on the standard assumptions underlying the pair approximation and diffusion reduction: weak selection, invasion by a single mutant, population size much larger than graph degree, and locally tree-like \(k\)-regular graphs. These assumptions delimit both the validity of the analytical results and the interpretation of the numerical comparisons. Simulations for which \(N\delta\) and \(k/N\) are sufficiently small are used for quantitative comparisons with the first-order formulas, whereas simulations outside this range are interpreted only as qualitative checks of the predicted trends. Strong selection, small populations, clustered or heterogeneous-degree networks, and initial conditions containing many mutants may generate dynamics and correlations that are not captured by the present approximation. Whether the same dilution mechanism persists under such conditions remains a natural question for future work.

\appendix

\section{\label{AppendixA}Pair Approximation and Diffusion Dynamics}
This appendix derives the one-dimensional diffusion approximation used in the main text. We begin with the pair identities that express all pair and conditional frequencies in terms of the global cooperator frequency \(p_C\) and the local correlation \(q_{C|C}\). Then, the derivation proceeds in five steps. First, we compute the focal and conditional payoffs under the locally tree-like pair approximation. Second, we expand the transition probabilities for \(p_C\) over one time step to first order in \(\delta\). Third, we derive the neutral quasi-steady relation for \(q_{C|C}\). Fourth, we substitute this relation into the slow equation for \(p_C\) to obtain a one-dimensional diffusion approximation. Finally, we use this diffusion process to derive the fixation probabilities.

Throughout this appendix, the pair approximation is used in the standard locally tree-like regime for regular graphs. More precisely, we assume \(3\le k\ll N\). In this regime, adjacent vertices are unlikely to share additional common neighbors, so local neighborhoods can be described by conditional pair frequencies to leading order. We also assume weak selection, \(\delta\ll1\), with the finite population scaling \(N\delta\ll1\) used when expanding the Kimura fixation formula. Thus, the results below should be interpreted as first-order approximations for large, sparse, locally tree-like \(k\)-regular graphs under weak selection with a single-mutant initial condition.

Let $p_C$ and $p_D = 1 - p_C$ be the global frequencies of cooperators and defectors in the population, respectively. For any two strategies $X, Y \in \{C, D\}$, we define $p_{XY}$ as the frequency of $X-Y$ neighbor pairs, and $q_{X|Y}$ as the conditional probability of finding an $X$-type neighbor given a $Y$-type focal individual. These quantities are linked by the balance equation 
$$p_{XY} = p_X q_{Y|X} = p_Y q_{X|Y},$$ 
which yields $q_{Y|X} = p_Y q_{X|Y} / p_X$. Furthermore, the conditional probabilities must normalize to unity, meaning $q_{C|X} + q_{D|X} = 1$. Consequently, the entire state of the system can be parameterized using the global cooperator frequency $p_C$ and the local cooperator-cooperator correlation $q_{C|C}$. The remaining probabilities are completely determined as follows
\begin{subequations}\label{all-probs}
\begin{align}
p_D &= 1 - p_C,\\
q_{D|C} &= 1 - q_{C|C},\\
q_{C|D} &= \frac{p_C q_{D|C}}{p_D} = \frac{p_C(1 - q_{C|C})}{1 - p_C},\\
q_{D|D} &= 1 - q_{C|D} = 1 - \frac{p_C(1 - q_{C|C})}{1 - p_C},\\
p_{CC} &= p_C q_{C|C},\\
p_{CD} &= p_C q_{D|C} = p_C(1 - q_{C|C}),\\
p_{DD} &= p_D q_{D|D} = (1 - p_C) - p_C(1 - q_{C|C}).
\end{align}
\end{subequations}
\subsection{\label{AppendixA-part1}Expected and Conditional Payoffs}
We first derive focal expected payoffs and then conditional payoffs in pairwise configurations. Throughout, we use the locally tree-like approximation with \(N\gg k\). Thus, a focal individual and one of its neighbors are unlikely to share additional common neighbors, and the remaining neighbors can be treated as conditionally independent to leading order.

Consider a focal cooperator with exactly \(n\) cooperative neighbors. The local configuration probability is
\[
\binom{k}{n}q_{C|C}^n(1-q_{C|C})^{k-n}.
\]
The focal participates in \(k+1\) games: one game organized by itself and \(k\) games organized by its neighbors. To average over the neighborhoods of these hosts, define, for \(X\in\{C,D\}\),
\[
\langle g(\ell)\rangle_{X,m}=\sum_{\ell=0}^m\binom{m}{\ell}q_{C|X}^{\ell}(1-q_{C|X})^{m-\ell}g(\ell).
\]
Thus,
\[
\pi_C^C=\langle a_{\ell+1}\rangle_{C,k-1},\qquad
\pi_C^D=\langle a_{\ell}\rangle_{D,k-1},
\]
where the shift $\ell+1$ accounts for the focal cooperator. Averaging over all $k+1$ games gives
\begin{equation}\label{payoffC}
\pi_C(n) = \frac{a_n + n\pi_C^C + (k-n)\pi_C^D}{k+1}.
\end{equation}
For a focal defector with $n$ cooperative neighbors,
\[
\binom{k}{n}q_{C\mid D}^n(1-q_{C\mid D})^{k-n}.
\]
The analogous neighborhood payoffs are
\[
\pi_D^C=\langle b_{\ell+1}\rangle_{C,k-1},\qquad
\pi_D^D=\langle b_{\ell}\rangle_{D,k-1},
\]
and the total expected payoff is
\begin{equation}\label{payoffD}
\pi_D(n)=\frac{b_n+n\pi_D^C+(k-n)\pi_D^D}{k+1}.
\end{equation}

We now compute conditional payoffs $\pi_{X|Y}$, i.e., the expected payoff of a focal type-$X$ individual given one fixed neighbor of type $Y$. Applying pair approximation to the remaining $k-1$ neighbors yields four cases:
\begin{subequations}
\begin{align}
\pi_{C|C}
&=\frac{\left(2+(k-1)q_{C|C}\right)\pi_C^C+(k-1)q_{D|C}\pi_C^D}{k+1},\\
\pi_{C|D}
&=\frac{\langle a_{\ell}\rangle_{C,k-1}+(k-1)q_{C|C}\pi_C^C+\left(1+(k-1)q_{D|C}\right)\pi_C^D}{k+1},\\
\pi_{D|D}
&=\frac{\left(2+(k-1)q_{D|D}\right)\pi_D^D+(k-1)q_{C|D}\pi_D^C}{k+1},\\
\pi_{D|C}
&=\frac{\langle b_{\ell+1}\rangle_{D,k-1}+(k-1)q_{D|D}\pi_D^D+\left(1+(k-1)q_{C|D}\right)\pi_D^C}{k+1}.
\end{align}
\end{subequations}
\subsection{\label{AppendixA-part3}$p_C$ in one time step}
In a single time step, the global frequency of cooperators, $p_C$, decreases by $1/N$ if a randomly selected cooperator updates its strategy to defection. This occurs when a focal cooperator decides to revise its strategy and successfully imitates one of its defector neighbors. The probability of this event is given by
\begin{align}
&\Pr\Bigl(\Delta p_C=-\frac{1}{N}\Bigr)=p_C \left\langle
\frac{1}{1+e^{\delta(\pi_C(n)-\alpha_C)}}
\frac{(k-n)e^{\delta\pi_{D|C}}}{
 ne^{\delta\pi_{C|C}}+(k-n)e^{\delta\pi_{D|C}}}\right\rangle_{C,k},\label{PrPC-}
\end{align}
where \(\langle \cdot \rangle_{C,k}\) denotes expectation with respect to \(n\), the number of cooperative neighbors of the focal cooperator.

Under weak selection, \(\delta\ll1\), we expand the two factors to first order
\begin{equation*}
\begin{split}
\frac{1}{1+e^{\delta(\pi_C(n)-\alpha_C)}}&=\frac{1}{2}-\frac{\delta}{4}\Bigl(\pi_C(n)-\alpha_C\Bigr)+\mathcal{O}(\delta^2),\\
\frac{(k-n)e^{\delta\pi_{D|C}}}{ne^{\delta\pi_{C|C}}+(k-n)e^{\delta\pi_{D|C}}}&=\frac{k-n}{k}+\delta\frac{(k-n)n}{k^2}\left(\pi_{D|C}-\pi_{C|C}\right)+\mathcal{O}(\delta^2).
\end{split}
\end{equation*}
Multiplying these expansions, using the payoff expressions from Subsection~\ref{AppendixA-part1}, and substituting into Eq.~\eqref{PrPC-}, we obtain
\begin{equation}\label{PrPC--}
\begin{split}
\Pr\Bigl(\Delta p_C=-\frac{1}{N}\Bigr)
&= \frac{p_C}{2k} \left\langle k-n \right\rangle_{C,k}+\frac{\delta p_C}{2k^2}\left(\pi_{D|C}-\pi_{C|C}\right)\left\langle(k-n)n \right\rangle_{C,k}+\frac{\delta p_C}{4k}\Biggl[\alpha_C\left\langle k-n \right\rangle_{C,k}\\
&\quad-\frac{\pi_C^C\left\langle n(k-n)\right\rangle_{C,k}}{k+1}
-\frac{\pi_C^D\left\langle (k-n)^2\right\rangle_{C,k}}{k+1}-\frac{\left\langle (k-n)a_n\right\rangle_{C,k}}{k+1}
\Biggr]+\mathcal{O}(\delta^2).
\end{split}
\end{equation}
Because $n$ follows a binomial distribution with parameters $k$ and $q_{C|C}$, we have the following moments
\begin{align*}
\langle k-n\rangle_{C,k} &= kq_{D|C},\\
\langle (k-n)^2\rangle_{C,k} &= kq_{D|C}\left(q_{C\mid C} + kq_{D\mid C}\right),\\
\langle n(k-n)\rangle_{C,k} &= k(k-1)q_{C\mid C}q_{D\mid C}.
\end{align*}
By substituting these moments into Eq.~\eqref{PrPC--} and using the relationship $p_Cq_{D|C}=p_{CD}$, we obtain 
\begin{align}
\Pr\Bigl(\Delta p_C=-\frac{1}{N}\Bigr)
&= \frac{p_{CD}}{2}
+\frac{\delta(k-1)p_{CD}q_{C|C}}{2k}\left(\pi_{D|C}-\pi_{C|C}\right)+\frac{\delta p_{CD}}{4}\Biggl[\alpha_C
-\frac{1}{k+1}\Bigl((k-1)q_{C|C}\pi_C^C\nonumber\\
&\qquad\qquad\qquad+(q_{C|C}+kq_{D|C})\pi_C^D\Bigr)-\frac{\left\langle (k-n)a_n\right\rangle_{C,k}}{k(k+1)q_{D|C}}\Biggr]+\mathcal{O}(\delta^2).\label{PrPC---}
\end{align}

Conversely, the frequency of cooperators increases by $1/N$ if a focal defector updates its strategy and imitates a cooperative neighbor. This occurs with probability
\begin{equation}\label{PrPC+}
\Pr\left(\Delta p_C=\frac{1}{N}\right)
= p_D \left\langle
\frac{1}{1+e^{\delta(\pi_D(n)-\alpha_D)}}
\frac{ne^{\delta\pi_{C|D}}}{
 ne^{\delta\pi_{C|D}}+(k-n)e^{\delta\pi_{D|D}}}\right\rangle_{D,k},
\end{equation}
where $n$ is now binomially distributed with parameters $k$ and $q_{C|D}$. Similarly, we have
\begin{align}
\Pr\Bigl(\Delta p_C = \frac{1}{N}\Bigr)
&= \frac{p_{CD}}{2}
 + \frac{\delta(k-1)p_{CD}q_{D|D}}{2k}\left(\pi_{C|D}-\pi_{D|D}\right)+ \frac{\delta p_{CD}}{4}\Biggl[\alpha_D
- \frac{1}{k+1}\Bigl((k-1)q_{D|D}\pi_D^D\nonumber\\
&\quad\quad\qquad+(q_{D|D}+kq_{C|D})\pi_D^C\Bigr) - \frac{\langle nb_n\rangle_{D,k}}{k(k+1)q_{C|D}}\Biggr] + \mathcal{O}(\delta^2).\label{PrPC+++}
\end{align}

Combining Eqs.~\eqref{PrPC---} and \eqref{PrPC+++}, the expected rate of change for $p_C$ can be written as
\begin{equation}\label{sec1_eq1}
\dot{p}_C = \frac{1}{N}\left[\Pr\Bigl(\Delta p_C = \frac{1}{N}\Bigr) - \Pr\Bigl(\Delta p_C = -\frac{1}{N}\Bigr)\right]= \frac{\delta p_{CD}}{N}\left(\frac{\Delta\alpha}{4} + F(p_C,q_{C|C})\right) + \mathcal{O}(\delta^2).
\end{equation}
Here, $F(p_C, q_{C|C})$ is a function capturing the local payoff differentials given by
\begin{equation}\label{sec1_eq2}
\begin{split}
&F(p_C,q_{C|C})\\
=&\frac{(k-1)}{2k}\left[q_{D|D}\left(\pi_{C|D}-\pi_{D|D}\right)+q_{C|C}\left(\pi_{C|C}-\pi_{D|C}\right)\right]+\frac{1}{4}\Biggl[
+\frac{\langle(k-n)a_n\rangle_{C,k}}{k(k+1)q_{D|C}}-\frac{\langle nb_n\rangle_{D,k}}{k(k+1)q_{C|D}}\\
&+\frac{(k-1)q_{C|C}\pi_C^C+(q_{C|C}+kq_{D|C})\pi_C^D-(k-1)q_{D|D}\pi_D^D-(q_{D|D}+kq_{C|D})\pi_D^C}{k+1}\Biggr]\\
=&\frac{(k-1)q_{C|C}\Bigl(3k+2+2(k-1)q_{C|C}-5kp_C\Bigr)}{4k(k+1)(1-p_C)}\pi_C^C+\frac{(k-1)q_{D|D}\langle a_{l}\rangle_{C,k-1}-(k-1)q_{C|C}\langle b_{l+1}\rangle_{D,k-1}}{2k(k+1)}\\
&+\frac{k(3k-2)(1-p_C) - 2k(k-1)p_C + (k-1)\big[ p_C(5k-2) - k \big]q_{C|C} - 2(k-1)^2 q_{C|C}^2}{4k(k+1)(1-p_C)}\pi_C^D\\
&-\frac{(k-1)\bigl(1-2p_C+p_Cq_{C|C}\bigr)\Bigl(3k+2+2(k-1)q_{C|C}-5kp_C\Bigr)}{4k(k+1)(1-p_C)^2}\pi_D^D+\frac{\langle(k-n)a_n\rangle_{C,k}}{4k(k+1)q_{D|C}}-\frac{\langle nb_n\rangle_{D,k}}{4k(k+1)q_{C|D}}\\
&-\frac{(3k-2)(1-p_C)^2 + (k-1)(1-p_C)\big[(5k-2)p_C-2\big](1-q_{C|C}) -2p_C(k-1)^2 (1-q_{C|C})^2 }{4k(k+1)(1-p_C)^2}\pi_D^C.
\end{split}
\end{equation}

\subsection{\label{AppendixA-part4}Frequency $q_{C|C}$}
Now, we are interested in the rate of change of $q_{C|C}$ in one time step. When an update results in a focal cooperator switching its strategy to defection, the total number of $C$-$C$ edges in the network decreases by $2n$, where $n$ is the number of cooperative neighbors connected to the focal cooperator. The probability of this specific event occurring in a given time step is
\begin{equation}\label{PrPCC-}
\begin{split}
    \Pr\left(\Delta p_{CC}=-\frac{2n}{kN}\right)
    &=p_C\frac{1}{1+e^{\delta(\pi_C(n)-\alpha_C)}}
    \frac{(k-n)e^{\delta\pi_{D|C}}}{
 ne^{\delta\pi_{C|C}}+(k-n)e^{\delta\pi_{D|C}}}=\frac{p_C(k-n)}{2k}+\mathcal{O}(\delta).
\end{split}
\end{equation}
Conversely, if the update results in a focal defector switching to cooperation, the number of $C$-$C$ edges increases by $2n$. The probability of this event is
\begin{align}
    &\Pr\left(\Delta p_{CC}=\frac{2n}{kN}\right)=p_D\frac{1}{1+e^{\delta(\pi_D(n)-\alpha_D)}}
\frac{ne^{\delta\pi_{C|D}}}{
 ne^{\delta\pi_{C|D}}+(k-n)e^{\delta\pi_{D|D}}}=\frac{np_D}{2k}+\mathcal{O}(\delta).\label{PrPCC+}
\end{align}

By combining Eqs.~\eqref{PrPCC-} and \eqref{PrPCC+} and taking the expectation over all possible neighborhood configurations, the expected rate of change of $p_{CC}$ in a single time step becomes
\begin{equation}\label{changePCC}
\begin{split}
\dot{p}_{CC}&=\left\langle\frac{2n}{kN} \Pr\left(\Delta p_{CC}=\frac{2n}{kN}\right)\right\rangle_{D,k}-\left\langle\frac{2n}{kN} \Pr\left(\Delta p_{CC}=-\frac{2n}{kN}\right)\right\rangle_{C,k}\\
&=\frac{p_D}{k^2N}\left\langle n^2\right\rangle_{D,k}-\frac{p_C}{k^2N}\left\langle n(k-n)\right\rangle_{C,k}+\mathcal{O}(\delta)\\
&=\frac{p_{CD}}{kN}\left(1+(k-1)(q_{C|D}-q_{C|C})\right)+\mathcal{O}(\delta).
\end{split}
\end{equation}
To isolate the dynamics of $q_{C|C}$, we use the relation $p_{CC} = q_{C|C}p_C$, which leads to
$$\dot{p}_{CC} = p_C\dot{q}_{C|C} + q_{C|C}\dot{p}_C.$$
Rearranging this expression allows us to solve for the rate of change of $q_{C|C}$ over one time step. Because $\dot{p}_C$ is of order $\mathcal{O}(\delta)$, the term $q_{C|C}\dot{p}_C/p_C$ gets absorbed into the $\mathcal{O}(\delta)$ remainder when evaluating the zeroth-order dynamics
\begin{equation}
\begin{split}
\dot{q}_{C|C}
&=\frac{\dot{p}_{CC}}{p_C}-\frac{q_{C|C}\dot{p}_C}{p_C}=\frac{q_{D|C}}{kN}\left(1+(k-1)(q_{C|D}-q_{C|C})\right)+\mathcal{O}(\delta).
\end{split}
\end{equation}

\subsection{\label{AppendixA-part5}Global Dynamics}
The evolutionary dynamics of the population unfold across two distinct time scales. On the fast time scale $\mathcal{O}(1)$, the local pair correlations, captured by $q_{C|C}$, relax rapidly to a quasi-steady state. In contrast, the global frequency of cooperators, $p_C$, evolves much more gradually on the slow time scale $\mathcal{O}(\delta)$. 
In this quasi-steady state, the expected change in local correlations vanishes ($\dot{q}_{C|C} = 0$), yielding the condition
\[
q_{C|C}=q_{C|D}+\frac{1}{k-1}.
\]

This condition is the neutral pair closure for locally tree-like regular graphs. It is obtained from the zeroth-order dynamics of local pair correlations, before payoff terms enter. This use is justified under weak selection because \(p_C\) changes on the slow \(O(\delta)\) time scale, while \(q_{C|C}\) relaxes on the fast neutral time scale. We therefore use this relation as the leading-order closure for the slow diffusion of \(p_C\). It is expected to be accurate for large, sparse, locally tree-like regular graphs, but it may fail when clustering, short loops, degree heterogeneity, strong selection, or many initial mutants create correlations not captured by pair frequencies.

By substituting this constraint back into Eq.~\eqref{all-probs}, we obtain
\begin{subequations}\label{identities}
\begin{align}
q_{C|C} &= \frac{1}{k-1}+\frac{k-2}{k-1}p_C,\\
p_D &= 1-p_C,\\
q_{D|C} &= \frac{k-2}{k-1}(1-p_C),\\
q_{C|D} &= \frac{k-2}{k-1}p_C,\\
q_{D|D} &= 1-\frac{k-2}{k-1}p_C,\\
p_{CD} &= \frac{k-2}{k-1}(1-p_C)p_C.
\end{align}
\end{subequations}
Thus, the entire dynamics of the population can be characterized by \(p_C\).

To model this evolution continuously, let us define a short time interval $\Delta t$ such that one unit of continuous time corresponds to $1/(\Delta t)$ discrete microscopic update steps. Within this diffusion framework, the infinitesimal mean of the change in $p_C$ per unit time is given by
\begin{equation}\label{mean}
\mathbb{E}[\Delta p_C|p_C(t)=p_C]=\dot{p}_C\Delta t=\delta \mu(p_C)\Delta t,
\end{equation}
where 
\begin{equation}
\mu(p_C)=\frac{p_{CD}}{N}\left(\frac{\Delta\alpha}{4}+G(p_C)\right),
\end{equation}
Here, we have
\begin{align}
  G(p_C)=& F\left(p_C, \frac{1}{k-1} + \frac{k-2}{k-1}p_C\right)\nonumber\\
  =&\frac{(3k + 4) + (3k + 4)(k - 2)p_C}{4k(k + 1)}\langle a_{n+1} \rangle_{C,k-1}+ \frac{(3k^2 - 3k - 2) - (k-2)(3k+2)p_C}{4k(k+1)}\langle a_{n} \rangle_{D,k-1}\nonumber\\
  &+ \frac{(k - 1) - (k - 2)p_C}{2k(k + 1)}\langle a_{n}\rangle_{C,k-1}+ \frac{(k-1)\langle(k-n)a_n\rangle_{C,k}}{4k(k+1)(k-2)(1-p_C)}
  -\frac{1 + (k-2)p_C}{2k(k+1)}\langle b_{n+1}\rangle_{D,k-1}\nonumber\\
  &-   \frac{(k-1)\langle nb_n\rangle_{D,k}}{4k(k+1)(k-2)p_C}-\frac{(3k + 4)(k - 1) - (3k + 4)(k - 2)p_C}{4k(k+1)}\langle b_{n}\rangle_{D,k-1}\nonumber\\
  & -\frac{(k+2) + (k-2)(3k+2)p_C}{4k(k+1)}\langle b_{n+1}\rangle_{C,k-1}.\label{expression_G}
\end{align}
The displayed expression is interpreted by continuous extension at $p_C=0$ and $p_C=1$; the apparent singularities are removable because the corresponding binomial moments vanish to the same order.
For infinitesimal variance of the change in $p_C$ in the quasi-steady state, a straightforward calculation shows that
\begin{equation}\label{variance}
\begin{split}
\mathrm{var}[\Delta p_C|p_C(t)=p_C]&=\frac{\Delta t}{N^2}\Bigl(\Pr\Bigl(\Delta p_C=\frac{1}{N}\Bigr)+\Pr\Bigl(\Delta p_C=-\frac{1}{N}\Bigr)\Bigr)+o(\Delta t)\\
&=\sigma^2(p_C)\Delta t+o(\Delta t),
\end{split}
\end{equation}
where
\begin{equation}
\sigma^2(p_C)=\frac{p_{CD}}{N^2}.
\end{equation}
As $\Delta t \ll 1$, Eqs.~\eqref{mean} and~\eqref{variance} show that the global frequency of cooperators can be approximated by a one-dimensional diffusion process with drift $\mu(p_C)$, variance $\sigma^2(p_C)$, and two absorbing boundaries at $p_C=0$ and $p_C=1$.

Now, suppose a proportion $p$ of cooperators is introduced into a resident population of defectors. The probability that this cooperator lineage eventually takes over the entire population is given by the classical Kimura diffusion formula \cite{crowkimura1970,ewens2004}
\begin{equation}\label{fixation}
\rho_C(p,\delta)\approx\frac{\int_{0}^p\exp\left\{-\int_0^x \frac{2\delta\mu(y)}{\sigma^2(y)}\,dy\right\}dx}{\int_{0}^1\exp\left\{-\int_0^x \frac{2\delta\mu(y)}{\sigma^2(y)}\,dy\right\}dx}.
\end{equation}
Dividing \eqref{mean} by \eqref{variance} leads to
\begin{equation}
\frac{\delta\mu(p_C)}{\sigma^2(p_C)} = N\delta\left(\frac{\Delta\alpha}{4}+G(p_C)\right).
\end{equation}
Hence
\[
\int_0^x\frac{2\delta\mu(y)}{\sigma^2(y)}\,dy=N\delta\left(\frac{x\Delta\alpha}{2}+2\int_0^xG(y)\,dy\right).
\]
Thus, the exponent in the fixation formula is controlled by \(N\delta\), not by \(\delta\) alone. If the payoff and aspiration parameters are of order one and \(G\) is bounded on \([0,1]\), then the exponent is uniformly of order \(N\delta\). Therefore, under the scaling \(N\delta\ll1\), the expansion
\[
e^{-u}=1-u+o(u)
\]
gives
\begin{equation}
\begin{split}
\int_{0}^{p}\exp\left\{-\int_0^x \frac{2\delta\mu(y)}{\sigma^2(y)}dy\right\}dx&=p-\delta\int_{0}^{p}\int_0^x \frac{2\mu(y)}{\sigma^2(y)}dydx+o(N\delta)\\
&=p-N\delta\left(\frac{p^2\Delta \alpha}{4}+2\int_{0}^{p}\int_0^xG(y)dydx\right)+o(N\delta).
\end{split}
\end{equation}
Evaluating the same expression at \(p=1\) gives the denominator in \eqref{fixation} as 
\begin{equation}
\int_{0}^{1}\exp\left\{-\int_0^x \frac{2\delta\mu(y)}{\sigma^2(y)}dy\right\}dx=1-N\delta\left(\frac{\Delta \alpha}{4}+2\int_{0}^{1}\int_0^xG(y)dydx\right)+o(N\delta).
\end{equation}
Then, the expansion of its reciprocal is also controlled by \(N\delta\).
Using the expansion 
\begin{equation}\label{eq_star}
\frac{1}{1-u}=1+u+o(u),
\end{equation}
under the assumption \(N\delta\ll1\), Eq.~\eqref{fixation} becomes
\begin{equation}\label{prob_C_p}
\begin{split}
\rho_C(p,\delta)
&=p+N\delta\left(\frac{p(1-p)\Delta \alpha}{4}+2p\int_{0}^{1}\int_0^xG(y)dydx-2\int_{0}^{p}\int_0^xG(y)dydx\right)+o(N\delta)\\
&=p+N\delta\left(\frac{p(1-p)\Delta \alpha}{4}+2(1-p)\int_{0}^{p} y\,G(y)dy + 2p\int_{p}^{1}(1-y)G(y)dy\right)+o(N\delta).
\end{split}
\end{equation}

By symmetry, if we introduce a proportion $p$ of defectors into a resident population of cooperators, the probability that the defectors fixate is simply
\begin{equation}\label{prob_D_p}
\begin{split}
\rho_D(p,\delta)
&=1-\rho_C(1-p,\delta)\\
&=p-N\delta\left(\frac{p(1-p)\Delta \alpha}{4}+2p\int_{0}^{1-p} y\,G(y)dy+ 2(1-p)\int_{1-p}^{1}(1-y)G(y)dy\right)+o(N\delta)
\end{split}
\end{equation}

Finally, taking the ratio of the two expressions in \eqref{prob_C_p} and \eqref{prob_D_p}, and using \eqref{eq_star},
we obtain
\begin{equation}\label{ratio_p}
\begin{split}
\frac{\rho_C(p,\delta)}{\rho_D(p,\delta)}
&=
1+N\delta \left(
\frac{(1-p)\Delta \alpha}{2}
+\frac{2}{p}\int_{0}^{1} \min\{y,1-y,p,1-p\} G(y) \, dy 
\right)
+o(N\delta).
\end{split}
\end{equation}
The reciprocal expansion is justified because the correction in the denominator is of order \(N\delta\). Thus, for fixed \(p>0\), the ratio expansion is valid when \(N\delta\ll1\), provided the payoff and aspiration parameters are of order one.

\subsection{\label{AppendixA-part6}Fixation Probability}

In this section, we will derive $\rho_C\!\left(\frac{1}{N},\delta\right)$, the fixation probability of cooperation introduced as a single individual in a population of defectors, and $\rho_D\!\left(\frac{1}{N},\delta\right)$, the fixation probability of defection introduced as a single individual in a population of cooperators. We begin by establishing the asymptotic behavior of the integral expressions in Eqs.~\eqref{prob_C_p}, \eqref{prob_D_p} and \eqref{ratio_p} as $p\ll1$.

\begin{lemma}\label{lem:rhoC_smallp}
Assume that $G$ is bounded. Then, for sufficiently small $p\ll 1$, we have
\begin{equation}\label{eq:rhoC_smallp}
(1-p)\int_{0}^{p} y\,G(y)\,dy+p\int_{p}^{1}(1-y)\,G(y)\,dy=p\int_{0}^{1}(1-y)\,G(y)\,dy+\mathcal{O}(p^{2}).
\end{equation}
\end{lemma}

\begin{proof}
First, note that
\[
\int_{p}^{1}(1-y)G(y)\,dy
=\int_{0}^{1}(1-y)G(y)\,dy-\int_{0}^{p}(1-y)G(y)\,dy.
\]
Substituting this expression into the left-hand side of Eq.~\eqref{eq:rhoC_smallp} yields
\[(1-p)\int_{0}^{p} yG(y)\,dy+p\int_{p}^{1}(1-y)G(y)\,dy=p\int_{0}^{1}(1-y)G(y)\,dy+R(p),\]
where 
\[
R(p)=(1-p)\int_{0}^{p} yG(y)\,dy-p\int_{0}^{p}(1-y)G(y)\,dy.
\]
Since $G$ is bounded, we let $M = \sup_{y \in [0,1]} |G(y)| < \infty$. It follows that
\begin{align*}
&\left|\int_{0}^{p} yG(y)\,dy\right|\le M\int_{0}^{p} y\,dy=\frac{M}{2}p^{2},\\
&\left|\int_{0}^{p} (1-y)G(y)\,dy\right|\le M\int_{0}^{p} 1\,dy=Mp.
\end{align*}
Consequently, the remainder satisfies
\[
|R(p)|
\le (1-p)\cdot\frac{M}{2}p^{2}+p\cdot Mp
=\mathcal{O}(p^{2}),
\]
which completes the proof. 
\end{proof}

Substituting the expansion from Eq.~\eqref{eq:rhoC_smallp} into Eq.~\eqref{prob_C_p} with $p=1/N$ yields the following result.
\begin{proposition}\label{prop:rhoC_asymp}
The fixation probability of a single cooperator mutant in a resident population of defectors is given by the approximation
\begin{equation}\label{eq:rhoC_coeff}
\rho_C
\approx\frac{1}{N}
+\delta\left[
\frac{\Delta\alpha}{4}
+2\int_0^1 (1-y)G(y)\,dy
\right].
\end{equation}
\end{proposition}

Using arguments similar to those used in the proof of Lemma~\ref{lem:rhoC_smallp}, we have the following asymptotic expansions as $p\ll1$:
\begin{subequations}
\begin{align}
p\int_{0}^{1-p} y\,G(y)\,dy+(1-p)\int_{1-p}^{1}(1-y)\,G(y)\,dy&=p\int_{0}^{1} y\,G(y)\,dy+\mathcal{O}(p^{2}),\label{eq:rhoD_smallp}\\
\int_0^1 \min\{y,1-y,p,1-p\}\,G(y)\,dy&=p\int_0^1 G(y)\,dy+\mathcal{O}(p^2).\label{eq:min_expansion}
\end{align}
\end{subequations} 
Inserting these expansions into Eqs.~\eqref{prob_D_p} and \eqref{ratio_p} with $p = 1/N$ leads to the following proposition.
\begin{proposition}
The fixation probability of a single defector mutant in a population of cooperators, and the corresponding ratio of the two fixation probabilities, are approximated as
\begin{subequations}
\begin{align}
\rho_D&\approx\frac1N-\delta\left[\frac{\Delta\alpha}{4}+2\int_0^1 y\,G(y)\,dy\right],\label{eq:rhoD_coeff}\\
\frac{\rho_C}{\rho_D}&\approx1+N\delta\left(\frac{\Delta\alpha}{2}+2\int_0^1 \,G(y)\,dy\right).\label{eq:rhoCrhoD_coeff}
\end{align}
\end{subequations}
\end{proposition}

\section{\label{AppendixB} Analysis for public goods game}
To analyze the evolutionary dynamics, we first derive the expected quantities in Eq.~\eqref{expression_G}. Using the moments of a binomial random variable with parameters 
$k-1$ or $k$ and $q_{C|C}=\frac{1}{k-1}+\frac{k-2}{k-1}p_C$, we obtain

\begin{equation}\label{public_a}
\langle a_{n+i} \rangle_{C,k-1}=\left\langle\frac{n+i+1}{k+1}b - c\right\rangle_{C,k-1}=\frac{(k-1)q_{C|C}+i+1}{k+1}b - c= \frac{i + 2 + (k-2)p_C}{k+1}b-c,
\end{equation}
for $i=0,1$.
By a similar argument, we have
\begin{subequations}\label{public_aa}
\begin{align}
\langle b_{n+1}\rangle_{C,k-1}&=\frac{2 + (k-2)p_C}{k+1}b,\\
\langle(k-n)a_n\rangle_{C,k}&=  \frac{k(k-2)(1-p_C)}{k-1}\Big[\frac{2 + (k-2)p_C}{k+1}b-c\Big],\\
\langle a_{n} \rangle_{D,k-1}&=
\frac{ 1 + p_C(k - 2)}{k + 1}b-c,\\
\langle b_{n+i}\rangle_{D,k-1}&=
 \frac{i + p_C(k - 2)}{k + 1}b,\\
\langle nb_n\rangle_{D,k}&=
\frac{k(k-2)p_C [1 + (k-2)p_C]}{(k-1)(k+1)}b,
\end{align}
\end{subequations}
for $i=0,1$. 
Inserting these expressions in Eq.~\eqref{expression_G}, we have 
\begin{equation}
G(p_C)=\frac{3(k+3)}{4(k+1)^2}b - \frac{3}{4}c.
\end{equation}
Inserting this expression in Eq.~\eqref{eq:rhoC_coeff}, and performing the required integrations, we obtain
\begin{subequations}
\begin{align}
\rho_C&\approx\frac{1}{N} + \delta \left( \frac{\Delta\alpha}{4} + \frac{3(k+3)}{4(k+1)^2}b - \frac{3}{4}c \right),\\
\rho_D&\approx\frac{1}{N} - \delta \left( \frac{\Delta\alpha}{4} + \frac{3(k+3)}{4(k+1)^2}b - \frac{3}{4}c \right),\\
\frac{\rho_C}{\rho_D}&\approx1 + N\delta \left( \frac{\Delta\alpha}{2} + \frac{3(k+3)}{2(k+1)^2}b - \frac{3}{2}c \right).
\end{align}
\end{subequations}
\section{\label{AppendixC} Analysis for the stag-hunt game}
In this section, we consider the stag-hunt game.
\subsection{\label{appendixC-part1}Fixation probability}
First, we need to evaluate the different payoffs in Eq.~\eqref{expression_G}. We have
\begin{subequations}\label{stag_hunt_a}
\begin{align}
\langle a_{n+1} \rangle_{C,k-1}&= q_{C\mid C}^{\,k-1} b - c,\\
\langle (k-n)a_n\rangle_{C,k}&= -k(1-q_{C\mid C})c,\\
\langle a_{n} \rangle_{D,k-1}&=\langle a_{n} \rangle_{C,k-1}=-c,\\
\langle b_{n+i}\rangle_{D,k-1}&=\langle b_{n+1}\rangle_{C,k-1}=\langle nb_n\rangle_{D,k}=0.
\end{align}
\end{subequations}
By substituting these expected payoffs in \eqref{stag_hunt_a} into \eqref{expression_G} and using \eqref{identities}, we have
\begin{equation}\label{eq:G_staghunt}
G(p_C)=\frac{(3k+4)\left(1+(k-2)p_C\right)^k}{4k(k+1)(k-1)^{k-1}}b-\frac{3}{4}c.
\end{equation}
Inserting this identity in \eqref{eq:rhoC_coeff} and integrating leads to
\begin{subequations}\label{eq:fix_sh_final}
\begin{align}
\rho_C
&\approx\frac1N+\frac{\delta}{4}\Bigg[\Delta\alpha+\;b\,\Phi_C^{\mathrm{SH}}(k)-3c
\Bigg],\\
\rho_D
&\approx\frac1N-\frac{\delta}{4}\Bigg[\Delta\alpha+\;b\,\Phi_D^{\mathrm{SH}}(k)-3c
\Bigg],\\
\frac{\rho_C}{\rho_D}
&\approx1+\frac{N\delta}{2}\Bigg[\Delta\alpha+\;b\,\Phi_R^{\mathrm{SH}}(k)-3c
\Bigg],
\end{align}
\end{subequations}
where 
\begin{align*}
\Phi_C^{\mathrm{SH}}(k)
    &=\frac{2(3k+4)((k-1)^{k+2}+3-k^2)}{k\,(k+1)^2\,(k+2)\,(k-2)^2(k-1)^{\,k-1}},\\
\Phi_D^{\mathrm{SH}}(k)
&=\frac{2(3k+4)((k-1)^{k+1}(k^2-k-3)+1)}{k\,(k+1)^2\,(k+2)\,(k-2)^2(k-1)^{\,k-1}},\\
\Phi_R^{\mathrm{SH}}(k)
&=\frac{(3k+4)((k-1)^{k+1}-1)}{k\,(k+1)^2\,(k-2)(k-1)^{\,k-1}}.
\end{align*}

\subsection{\label{appendixC-part2} $(\Phi_S^{\mathrm{SH}}(k))_{k\geq3}$ is decreasing for $S=C,D,R$}

First, we will show that $(\Phi_C^{\mathrm{SH}}(k))_k$ is decreasing. A direct algebraic rearrangement gives
\begin{equation}\label{eq:Tk_factor}
\Phi_C^{\mathrm{SH}}(k)=2\,U_1(k)\bigl(1-\varepsilon_1(k)\bigr),
\end{equation}
where
\[
U_1(k)=\frac{(3k+4)(k-1)^3}{k(k+1)^2(k+2)(k-2)^2},
\qquad
\varepsilon_1(k)=\frac{k^2-3}{(k-1)^{k+2}}.
\]

\begin{lemma}\label{lem:U_decreasing}
For all integers $k \ge 3$, we have
\begin{equation}\label{eq:U_ratio_bound}
\frac{U_1(k+1)}{U_1(k)}< \frac{k+1}{k+3}
=1-\frac{2}{k+3}.
\end{equation}
\end{lemma}

\begin{proof}
Extend $U_1(x)$ to any real $x>2$. Note that
\[
\frac{d}{dx}\log U_1(x)=
\frac{3}{3x+4}+\frac{3}{x-1}-\frac{1}{x}
-\frac{2}{x+1}-\frac{1}{x+2}-\frac{2}{x-2}.
\]
In addition, we have
\[
\frac{3}{3x+4}\le \frac{1}{x+1},
\qquad
\frac{1}{x-1}\le \frac12\left(\frac{1}{x-2}+\frac{1}{x}\right),
\]
where the second inequality follows from the convexity of $t\mapsto 1/t$.
Substituting these inequalities into the derivative yields
\[
\frac{d}{dx}\log U_1(x)
<-\frac{1}{x+1}-\frac{1}{x+2}.
\]
Integrating from $k$ to $k+1$ gives
\[
\log\left(\frac{U_1(k+1)}{U_1(k)}\right)
<-\int_k^{k+1}\left(\frac{1}{x+1}+\frac{1}{x+2}\right)\,dx
=\log\left(\frac{k+1}{k+3}\right).
\]
Exponentiating the last inequality proves the claim.
\end{proof}

\begin{lemma}\label{lem:eps_control}
For all integers $k \ge 4$, the correction factor $\varepsilon_1(k)$ satisfies
\begin{equation}\label{eq:eps_bound}
\frac{1-\varepsilon_1(k+1)}{1-\varepsilon_1(k)}
\le 1+\frac{2}{k+3}.
\end{equation}
\end{lemma}

\begin{proof}
For $k\ge4$, we have
\[
\varepsilon_1(k)=\frac{k^2-3}{(k-1)^{k+2}}
\le \frac{k^2}{3^{k+2}}.
\]
On the other hand, by induction we obtain $3^{k+2}\ge k^2(k+3)$ for $k\ge4$.
Therefore, we get
\[
\varepsilon_1(k)\le \frac{1}{k+3},
\]
in particular, $\varepsilon_1(k)\le 1/7<1/2$. Now, using the inequality 
$1/(1-x)\le 1+2x$ for $0\le x\le 1/2$, we obtain
\[
\frac{1-\varepsilon_1(k+1)}{1-\varepsilon_1(k)}
\le \frac{1}{1-\varepsilon_1(k)}
\le 1+2\varepsilon_1(k)
\le 1+\frac{2}{k+3}.
\]
\end{proof}

\begin{proposition}\label{prop:T1_decreasing}
The sequence $(\Phi_C^{\mathrm{SH}}(k))_k$ is strictly decreasing for all integers
$k\ge3$.
\end{proposition}

\begin{proof}
From Eq.~\eqref{eq:Tk_factor}, we have
\[
\frac{\Phi_C^{\mathrm{SH}}(k+1)}{\Phi_C^{\mathrm{SH}}(k)}
=\frac{U_1(k+1)}{U_1(k)}
\frac{1-\varepsilon_1(k+1)}{1-\varepsilon_1(k)}.
\]
For $k\ge4$, Lemmas~\ref{lem:U_decreasing} and~\ref{lem:eps_control} give
\[
\frac{\Phi_C^{\mathrm{SH}}(k+1)}{\Phi_C^{\mathrm{SH}}(k)}
<\left(1-\frac{2}{k+3}\right)\left(1+\frac{2}{k+3}\right)
=1-\frac{4}{(k+3)^2}<1,
\]
for $k\geq4$.
The remaining initial comparison is direct:
\[
\Phi_C^{\mathrm{SH}}(3)=\frac{169}{240}
>\Phi_C^{\mathrm{SH}}(4)=\frac{716}{2025}.
\]
Thus $(\Phi_C^{\mathrm{SH}}(k))_{k\geq3}$ is strictly decreasing.
\end{proof}

Now, we will prove the corresponding result for $\Phi_D^{\mathrm{SH}}$. Again,
note the factorization
\begin{equation}\label{eq:S1_factor}
\Phi_D^{\mathrm{SH}}(k)=2\,U_2(k)\,\bigl(1+\varepsilon_2(k)\bigr),
\end{equation}
where
\[
U_2(k)=
\frac{(3k+4)(k^2-k-3)(k-1)^2}
{k(k+1)^2(k+2)(k-2)^2},
\qquad
\varepsilon_2(k)=
\frac{1}{(k^2-k-3)(k-1)^{k+1}}.
\]

\begin{lemma}\label{lem:U2_decreasing}
We have $U_2(k+1)<U_2(k)$ for all integers $k\ge4$.
\end{lemma}

\begin{proof}
For real $x\geq4$, differentiating $\log U_2(x)$ gives
\[
\frac{d}{dx}\log U_2(x)=
\frac{3}{3x+4}
+\frac{2x-1}{x^2-x-3}
+\frac{2}{x-1}
-\frac{1}{x}
-\frac{2}{x+1}
-\frac{1}{x+2}
-\frac{2}{x-2}.
\]
Using
\[
\frac{3}{3x+4}\le \frac{1}{x+1},
\qquad
\frac{2}{x-1}\le \frac{1}{x}+\frac{1}{x-2},
\]
we obtain
\[
\frac{d}{dx}\log U_2(x)
\le
\frac{2x-1}{x^2-x-3}
-\frac{1}{x+1}
-\frac{1}{x+2}
-\frac{1}{x-2}.
\]
For $x\ge4$, a common-denominator calculation gives
\[
\frac{1}{x-2}+\frac{1}{x+1}+\frac{1}{x+2}
-\frac{2x-1}{x^2-x-3}
=
\frac{\frac{x^3}{2}(x-4)+\frac{x^2}{2}(x^2-12)+2x+8}
{(x-2)(x+1)(x+2)(x^2-x-3)} >0.
\]
Hence $d(\log U_2(x))/dx<0$ on $[4,\infty)$, and $U_2(k)$ is strictly
decreasing for integers $k\ge4$.
\end{proof}

\begin{lemma}\label{lem:eps2_decreasing}
The sequence $\varepsilon_2(k)$ is strictly decreasing for all integers
$k\ge4$.
\end{lemma}

\begin{proof}
This is a consequence of 
\[
\frac{\varepsilon_2(k+1)}{\varepsilon_2(k)}
=
\frac{k^2-k-3}{k^2+k-3}
\left(\frac{k-1}{k}\right)^{k+1}\frac{1}{k}<1.
\]
\end{proof}

\begin{proposition}\label{prop:T2_decreasing}
The sequence $\Phi_D^{\mathrm{SH}}(k)$ is strictly decreasing for all integers
$k\ge3$.
\end{proposition}

\begin{proof}
For $k\ge4$, Lemmas~\ref{lem:U2_decreasing} and~\ref{lem:eps2_decreasing}
show that both $U_2(k)$ and $1+\varepsilon_2(k)$ are positive and strictly
decreasing. Hence their product in Eq.~\eqref{eq:S1_factor} is strictly
decreasing for $k\ge4$. The initial comparison is
\[
\Phi_D^{\mathrm{SH}}(3)=\frac{637}{480}
>\Phi_D^{\mathrm{SH}}(4)=\frac{2188}{2025}.
\]
Therefore $\Phi_D^{\mathrm{SH}}(k)$ is strictly decreasing for all $k\ge3$.
\end{proof}

\begin{proposition}\label{prop:T3_decreasing}
The sequence $\Phi_R^{\mathrm{SH}}(k)$ is strictly decreasing for all integers
$k\ge3$.
\end{proposition}

\begin{proof}
Using the identity
\[
\Phi_R^{\mathrm{SH}}(k)=
\frac{\Phi_C^{\mathrm{SH}}(k)+\Phi_D^{\mathrm{SH}}(k)}{2},
\]
the result follows immediately from Propositions~\ref{prop:T1_decreasing}
and~\ref{prop:T2_decreasing}.
\end{proof}


\section{\label{AppendixD} Analysis for snowdrift game}

In this section, we consider the snowdrift game.
\subsection{\label{AppendixD-part1}Expression of $G$}
Let $n\sim\mathrm{Bin}(k-1,q_{C|C})$. 
Since $n$ is binomial, its probability generating function gives
\[
\langle t^n\rangle_{C,k-1}=\sum_{n=0}^{k-1}\binom{k-1}{n}q_{C|C}^{n}q_{D|C}^{k-1-n}t^{n}
=(q_{D|C}+q_{C|C}t)^{k-1}.
\]
Using the integral identity
\(
\frac{1}{m}=\int_0^1 t^{m-1}\,dt,
\)
for $m\ge 1$, we have
\begin{equation}\label{eq:Elplus1_inv}
\begin{split}
\left\langle\frac{1}{n+1}\right\rangle_{C,k-1}
=\int_0^1 \langle t^n\rangle_{C,k-1}\,dt
&=\int_0^1 (q_{D|C}+q_{C|C}t)^{k-1}\,dt
=\frac{1-q_{D|C}^{k}}{kq_{C|C}}.
\end{split}
\end{equation}
Similarly, we obtain
\begin{subequations}\label{snow-eq1}
\begin{align}
\langle a_n\rangle_{C,k-1}&=b-\frac{c}{k}\frac{1-q_{D|C}^{k}}{q_{C|C}},\\
\langle a_{n+1} \rangle_{C,k-1}&=b-c\,\frac{(k+1)q_{C|C}-1+q_{D|C}^{k+1}}{k(k+1)q_{C|C}^2},\\
\langle b_{n+1} \rangle_{C,k-1}&=\langle b_{n+1}\rangle_{D,k-1}=b,\\
\langle(k-n)a_n \rangle_{C,k}&=kbq_{D|C}-c\frac{q_{D|C}(1-q_{D|C}^k)}{q_{C|C}},\\
\langle a_n\rangle_{D,k-1}&=b-c\frac{1-q_{D|D}^k}{kq_{C|D}},\\
\langle b_n\rangle_{D,k-1}&=b\Big(1-q_{D|D}^{k-1}\Big),\\
\langle n b_n\rangle_{D,k}&=bkq_{C|D}.
\end{align}
\end{subequations}

Substituting the quasi-steady state identities Eq.~\eqref{identities} into these expected payoffs and then inserting the result into Eq.~\eqref{expression_G} gives
\begin{equation}\label{eq:G_snowdrift_affine}
G(y)=G_b(y)\,b-G_c(y)\,c,
\end{equation}
where the benefit coefficient is
\begin{equation}\label{eq:phi_b_snowdrift}
G_b(y)=\frac{3k+4}{4k(k+1)}\,
\frac{\big(k-1-(k-2)y\big)^{k}}{(k-1)^{\,k-1}},
\end{equation}
and the cost coefficient is
\begin{align}\label{eq:phi_c_snowdrift}
G_c(y)
=&\frac{(3k+4)(k-1)}{4k^2(k+1)^2}\frac{2+(k+1)\beta y+(k-1)a^{k+1}(1-y)^{k+1}}{1+\beta y}+\frac{(k-1)(3k^2-3k-2)}{4k^2(k+1)\beta}\frac{1-(1-a y)^k}{y}\nonumber\\
& +\frac{(k-1)(3k-2-2\beta y)}{4k^2(k+1)(1+\beta y)}
\left(1-a^k(1-y)^k\right)-\frac{(k-1)(3k+2)}{4k^2(k+1)}\big(1-(1-a y)^k\big).
\end{align}
Here we introduce the shorthands $\beta=k-2$, $a=\tfrac{k-2}{k-1}$ and $L=\ln(k-1)$. 

\subsection{\label{appendixD-part2}Fixation probability}
Inserting \eqref{eq:phi_c_snowdrift} in \eqref{eq:rhoC_coeff} and integrating the different factors, we obtain 
\begin{subequations}
\begin{align}
\rho_C
&\approx\frac1N+\frac{\delta}{4}\Bigg[\Delta\alpha
+\;b\,\Phi_C^{\mathrm{SD}}(k)\;-\;c\,\Gamma_{C}^{\mathrm{SD}}(k)\Bigg],\\
\rho_D
&\approx\frac1N-\frac{\delta}{4}\Bigg[\Delta\alpha
+\;b\,\Phi_D^{\mathrm{SD}}(k)\;-\;c\,\Gamma_{D}^{\mathrm{SD}}(k)\Bigg],\\
\frac{\rho_C}{\rho_D}
&\approx1+\frac{N\delta}{2}\Bigg[\Delta\alpha
+\;b\,\Phi_R^{\mathrm{SD}}(k)\;-\;c\,\Gamma_{R}^{\mathrm{SD}}(k)\Bigg],
\end{align}
\end{subequations}
where
\begin{subequations}\label{snowdrift-phi}
\begin{align}
\Phi_C^{\mathrm{SD}}(k)
=&\frac{2(3k+4)((k-1)^{k+1}(k^2-k-3)+1)}{k(k-2)^2(k+1)^{2}(k+2)(k-1)^{k-1}},\\
\Phi_D^{\mathrm{SD}}(k)
=&\frac{2(3k+4)((k-1)^{k+2}+3-k^2)}{k(k-2)^2(k+1)^2(k+2)(k-1)^{k-1}},\\
\Phi_R^{\mathrm{SD}}(k)
=&\frac{(3k+4)((k-1)^{k+1}-1)}{k(k-2)(k+1)^2(k-1)^{k-1}}.
\end{align}
\end{subequations}

For the coefficients of $c$, first we need to introduce some special functions.
Note the integral representation
\begin{equation}\label{eq:beta_hypergeom_identity}
\begin{split}
\quad\int_0^1 \frac{t^{r}(1-t)^{s}}{1+\beta t}\,dt
&=\frac{r!s!}{(r+s+1)!}\,{}_2F_1\!\Big(1,r+1;\,r+s+2;\,-\beta\Big),
\end{split}
\end{equation}
where ${}_2F_1$ is the Gauss hypergeometric function. We define the hypergeometric shorthand $H_{m,n}(k):={}_2F_1\!\Big(1,m;n;2-k\Big)$, noting that $2-k=-\beta$. Next, we define auxiliary quantities to handle the integration of terms containing the polynomial $(1-a y)^k$ as
\begin{align}
I_1(k)
&:=\int_0^1\big(1-(1-a y)^k\big)\,dy=
\frac{ (k-1)^k(k^2-2k-1)+1}{(k-1)^k\,(k-2)(k+1)},
\\
I_2(k)
&:=\int_0^1 y\,\big(1-(1-a y)^k\big)\,dy
=
\frac{ (k-1)^k(k^4-k^3-8k^2+8k+6)+2k^2-6}{2 (k-1)^k (k-2)^2 (k+1)(k+2)},
\\
S_k(a)
&:=\int_0^1\frac{1-(1-a y)^k}{y}\,dy
=
\sum_{j=1}^{k}(-1)^{j+1}\binom{k}{j}\frac{a^j}{j}.
\label{eq:Sk_def}
\end{align}
Finally, by collecting all terms in $G_c$, integrating term by term, and using identity~\eqref{eq:beta_hypergeom_identity}, we obtain
\begin{subequations}\label{eq:I1I2_final}
\begin{align}
\Gamma_C^{\mathrm{SD}}(k)=&\frac{2(3k+4)(k-1)}{k^2(k+1)^2}\left[\frac{k+1}{2}-(k-1)\frac{(k-1)L+2-k}{(k-2)^2}+\frac{(k-1)a^{k+1}}{k+3}H_{1,k+4}(k)\right]\nonumber\\
&+\frac{2(k-1)}{k^2(k+1)}\left[\frac{2a^k}{k+2}-1-3k\!\left(\frac{a^k}{k+2} H_{1,k+3}(k)-\frac{(k-1)L+2-k}{(k-2)^2}\right)\right]\label{eq:I1_final}\\
&+\frac{2(k-1)}{k^2(k-2)(k+1)}\Big[(3k^2-3k-2)\,\big(S_k(a)-I_1(k)\big)+(k-2)(3k+2)\,\big(I_2(k)-I_1(k)\big)\Big],\nonumber\\
\Gamma_D^{\mathrm{SD}}(k)
=&\frac{2(3k+4)(k-1)}{k^2(k+1)^2}\left[\frac{k+1}{2}-(k-1)\!\left(\frac{1}{k-2}-\frac{L}{(k-2)^2}\right)+\frac{(k-1)a^{k+1}}{(k+3)(k+2)}H_{2,k+4}(k)\right]\nonumber\\
&+\frac{2(k-1)}{k^2(k+1)}\left[-1+3k\!\left(\frac{1}{k-2}-\frac{L}{(k-2)^2}\right)-a^k\!\left(\frac{H_{2,k+3}(k)}{k+1}+\frac{2(k-2)}{(k+2)(k+3)}\,H_{2,k+4}(k)\right)\right]\nonumber\\
&+\frac{2(k-1)}{k^2(k-2)(k+1)}\Big[(3k^2-3k-2)\,I_1(k)-(k-2)(3k+2)\,I_2(k)\Big].
\end{align}
\end{subequations}
For the last coefficient, we have
\begin{equation}\label{eq:I3_final}
\Gamma_R^{\mathrm{SD}}(k)=\frac{\Gamma_C^{\mathrm{SD}}(k)+\Gamma_D^{\mathrm{SD}}(k)}{2}.
\end{equation}


\begin{remark}\label{snowdrift_coefficient}
The closed forms in \eqref{snowdrift-phi} satisfy
\(
\Phi_C^{\mathrm{SD}}(k)=\Phi_D^{\mathrm{SH}}(k)\)
and
\(\Phi_D^{\mathrm{SD}}(k)=\Phi_C^{\mathrm{SH}}(k)
\)
exactly. Consequently, the monotonicity of the sequences
$\bigl(\Phi_S^{\mathrm{SD}}(k)\bigr)_{k\ge 3}$ follows immediately from the corresponding
monotonicity results for $\bigl(\Phi_S^{\mathrm{SH}}(k)\bigr)_{k\ge 3}$.
\end{remark}

\subsection{\label{appendixD-part3}Large-$k$ asymptotics of $\Gamma_S^{\mathrm{SD}}(k)$}
The dominant leading-order contribution to the cost gradient as $k\to\infty$ is
\[
G_c(y)\approx \frac{3}{4k}\,\frac{1-(1-y)^k}{y}.
\]
Substituting into Eq.~\eqref{eq:I1_final} gives
\[
\Gamma^{\mathrm{SD}}_C \approx \frac{6}{k}\int_0^1 \frac{1-y}{y}\Bigl(1-(1-y)^k\Bigr)\,dy.
\]
Using the substitution $u=1-y$ and the identity
\(
\frac{1-u^k}{1-u}=\sum_{j=0}^{k-1}u^j,
\)
we obtain
\begin{equation}
\Gamma^{\mathrm{SD}}_C \approx \frac{6}{k}\sum_{j=1}^{k}\frac{1}{j+1} \sim \frac{6\ln k}{k}
\end{equation}
as $k\to\infty$. Similarly, we obtain
\begin{equation}
\Gamma^{\mathrm{SD}}_D(k)\sim\frac{6}{k}.
\end{equation}
Consequently, their midpoint evaluates to
\begin{equation}
\Gamma^{\mathrm{SD}}_R(k)\sim\frac{3\ln k}{k}.
\label{eq:Iavg_asymp}
\end{equation}

\section{\label{app:two_intensity}Two selection intensities}
The main text uses a single weak-selection parameter, \(\delta\), for both parts of the update rule. This means that aspiration responsiveness and imitation bias are assumed to operate on the same scale. Here, we give the first-order calculations for the more general case in which these two effects have different intensities. This shows which conclusions depend on the symmetric intensity assumption and which conclusions remain unchanged when the two strengths are separated.

Let \(\delta_A>0\) denote the intensity of aspiration responsiveness, and let \(\delta_I>0\) denote the intensity of imitation based on payoffs. The revision probability is
\begin{equation}\label{eq:two_delta_revision}
\Pr(i \;\text{revises})=
\frac{1}{1+\exp\{\delta_A(\pi_i-\alpha_i)\}},
\end{equation}
while, conditional on revision, the probability that individual \(i\) copies neighbor \(j\in\mathcal N(i)\) is
\begin{equation}\label{eq:two_delta_imitation}
\Pr(i\text{ imitates }j\mid i\text{ revises })=
\frac{\exp(\delta_I\pi_j)}
{\sum_{\ell\in\mathcal N(i)}\exp(\delta_I\pi_\ell)}.
\end{equation}

Assuming $\delta_A \ll 1$ and $\delta_I \ll 1$, we have
\begin{equation}\label{eq:gate_expansion_two_delta}
\frac{1}{1+\exp\{\delta_A(\pi_i-\alpha_i)\}}=\frac12-\frac{\delta_A}{4}(\pi_i-\alpha_i)+\mathcal{O}(\delta_A^2),
\end{equation}
and
\begin{equation}\label{eq:imit_expansion_two_delta}
\frac{\exp(\delta_I\pi_j)}{\sum_{\ell\in\mathcal N(i)}\exp(\delta_I\pi_\ell)}=\frac1k+\frac{\delta_I}{k}\left(\pi_j-\frac1k\sum_{\ell\in\mathcal N(i)}\pi_\ell\right)+\mathcal{O}(\delta_I^2).
\end{equation}
The pair-approximation payoffs from Appendix~\ref{AppendixA-part1} remain unchanged. The frequency of cooperation increases by $1/N$ with probability
\begin{align*}
\Pr\!\left(\Delta p_C = \frac{1}{N}\right) &= \frac{p_{CD}}{2} + \frac{\delta_I (k-1)p_{CD}q_{D\mid D}}{2k}\left(\pi_{C\mid D}-\pi_{D\mid D}\right) \\
&\quad + \frac{\delta_A p_{CD}}{4}\left[\alpha_D - \frac{(k-1)q_{D\mid D}\pi_D^D+(q_{D\mid D}+kq_{C\mid D})\pi_D^C}{k+1} - \frac{\langle n b_n\rangle_{D,k}}{k(k+1)q_{C\mid D}}\right] \\
&\quad + \mathcal{O}(\delta_A^2, \delta_I^2, \delta_A\delta_I),
\end{align*}
and decreases by $1/N$ with probability
\begin{align}
\Pr\!\left(\Delta p_C = -\frac{1}{N}\right) &= \frac{p_{CD}}{2} + \frac{\delta_I (k-1)p_{CD}q_{C\mid C}}{2k}\left(\pi_{D\mid C}-\pi_{C\mid C}\right) \nonumber \\
&\quad + \frac{\delta_A p_{CD}}{4}\left[\alpha_C - \frac{(k-1)q_{C\mid C}\pi_C^C+(q_{C\mid C}+kq_{D\mid C})\pi_C^D}{k+1} - \frac{\langle (k-n)a_n\rangle_{C,k}}{k(k+1)q_{D\mid C}}\right] \nonumber \\
&\quad + \mathcal{O}(\delta_A^2, \delta_I^2, \delta_A\delta_I).
\end{align}
The resulting drift is
\begin{equation}
\dot{p}_C= \frac{p_{CD}}{N} \left[ \delta_A\left(\frac{\Delta\alpha}{4} + F_A(p_C,q_{C\mid C})\right) + \delta_I\,F_I(p_C,q_{C\mid C}) \right] + \mathcal{O}(\delta_A^2,\delta_I^2,\delta_A\delta_I),
\end{equation}
where
\begin{align*}
F_I(p_C,q_{C\mid C}) &= \frac{k-1}{2k}\left[q_{D\mid D}\bigl(\pi_{C\mid D}-\pi_{D\mid D}\bigr)+q_{C\mid C}\bigl(\pi_{C\mid C}-\pi_{D\mid C}\bigr)\right],\\
F_A(p_C,q_{C\mid C}) &= \frac{1}{4}\Bigg[\frac{\langle (k-n)a_n\rangle_{C,k}}{k(k+1)q_{D\mid C}}-\frac{\langle n b_n\rangle_{D,k}}{k(k+1)q_{C\mid D}} \\
&\quad + \frac{(k-1)q_{C\mid C}\pi_C^C+(q_{C\mid C}+kq_{D\mid C})\pi_C^D-(k-1)q_{D\mid D}\pi_D^D-(q_{D\mid D}+kq_{C\mid D})\pi_D^C}{k+1}\Bigg].
\end{align*}

In the quasi-steady state from \eqref{identities}, the term $\delta\mu(p_C)$ becomes
\begin{equation}
M(p_C)=\frac{p_{CD}}{N}\left[\delta_A\left(\frac{\Delta\alpha}{4}+G_A(p_C)\right)+\delta_I G_I(p_C)\right],
\end{equation}
with
\begin{equation}\label{GA}
\begin{aligned}
G_A(p_C) &= \frac{1+(k-2)p_C}{4(k+1)}\langle a_{n+1}\rangle_{C,k-1} + \frac{k-1-(k-2)p_C}{4(k+1)}\langle a_{n}\rangle_{D,k-1} \\
&\quad -\frac{k-1-(k-2)p_C}{4(k+1)}\langle b_{n}\rangle_{D,k-1} - \frac{1+(k-2)p_C}{4(k+1)}\langle b_{n+1}\rangle_{C,k-1} \\
&\quad +\frac{(k-1)\langle (k-n)a_n\rangle_{C,k}}{4k(k+1)(k-2)(1-p_C)} - \frac{(k-1)\langle n b_n\rangle_{D,k}}{4k(k+1)(k-2)p_C},
\end{aligned}
\end{equation}
and
\begin{equation}\label{GI}
\begin{aligned}
G_I(p_C) &= \frac{(k+2)\bigl(1+(k-2)p_C\bigr)}{2k(k+1)}\langle a_{n+1}\rangle_{C,k-1} + \frac{k^2-k-1-(k^2-k-2)p_C}{2k(k+1)}\langle a_{n}\rangle_{D,k-1} \\
&\quad +\frac{k-1-(k-2)p_C}{2k(k+1)}\langle a_{n}\rangle_{C,k-1} - \frac{1+(k-2)p_C}{2k(k+1)}\langle b_{n+1}\rangle_{D,k-1} \\
&\quad -\frac{(k+2)\bigl(k-1-(k-2)p_C\bigr)}{2k(k+1)}\langle b_{n}\rangle_{D,k-1} - \frac{1+(k^2-k-2)p_C}{2k(k+1)}\langle b_{n+1}\rangle_{C,k-1}.
\end{aligned}
\end{equation}
Following the same logic as in Appendix \ref{AppendixA}, and assuming $N\delta_A \ll 1$ and $N\delta_I \ll 1$, we obtain the following approximations:
\begin{subequations}\label{eq:rho_two_delta}
\begin{align}
\rho_C &\approx \frac{1}{N} + \delta_A\left[\frac{\Delta\alpha}{4}+2\int_0^1 (1-y)G_A(y)\,dy\right] + \delta_I\left[2\int_0^1 (1-y)G_I(y)\,dy\right], \label{eq:rhoC_two_delta}\\
\rho_D &\approx \frac{1}{N} - \delta_A\left[\frac{\Delta\alpha}{4}+2\int_0^1 yG_A(y)\,dy\right] - \delta_I\left[2\int_0^1 yG_I(y)\,dy\right], \label{eq:rhoD_two_delta}\\
\frac{\rho_C}{\rho_D} &\approx 1 + N\left[\delta_A\left(\frac{\Delta\alpha}{2}+2\int_0^1 G_A(y)\,dy\right) + \delta_I\left(2\int_0^1 G_I(y)\,dy\right)\right]. \label{eq:ratio_two_delta}
\end{align}
\end{subequations}

\subsection{\label{appendixE-part1}Public goods game}

For the public goods game, we find
\begin{equation}
G_A = \frac{b}{4(k+1)}-\frac{c}{4}, \qquad G_I = \frac{(k+4)b}{2(k+1)^2}-\frac{c}{2}. \label{eq:PGG_GA_GI}
\end{equation}
Substituting these into Eq.~\eqref{eq:rho_two_delta} and integrating yields
\begin{subequations}
\begin{align}
\rho_C &\approx \frac{1}{N} + \frac{\delta_A}{4}\left(\Delta\alpha+\frac{b}{k+1}-c\right) + \frac{\delta_I}{4}\left(\frac{2(k+4)}{(k+1)^2}b-2c\right), \label{eq:PGG_rhoC_two_delta}\\
\rho_D &\approx \frac{1}{N} - \frac{\delta_A}{4}\left(\Delta\alpha+\frac{b}{k+1}-c\right) - \frac{\delta_I}{4}\left(\frac{2(k+4)}{(k+1)^2}b-2c\right), \label{eq:PGG_rhoD_two_delta}\\
\frac{\rho_C}{\rho_D} &\approx 1 + \frac{N}{2}\left[\delta_A\left(\Delta\alpha+\frac{b}{k+1}-c\right) + 2\delta_I\left(\frac{k+4}{(k+1)^2}b-c\right)\right]. \label{eq:PGG_ratio_two_delta}
\end{align}
\end{subequations}
The success criteria are equivalent and given by
\begin{equation}
\Delta\alpha + \frac{k+1+2\lambda(k+4)}{(k+1)^2}b > (1+2\lambda)c, \label{eq:PGG_condition_two_delta}
\end{equation}
where $\lambda=\delta_I/\delta_A$. If $\delta_A = \delta_I = \delta$, this reduces to the condition 
\[
\Delta\alpha + 3\left[\frac{k+3}{(k+1)^2}b-c\right] > 0.
\]
In the limit $k \to \infty$, the benefit term vanishes and Eq.~\eqref{eq:PGG_condition_two_delta} becomes
\begin{equation}
\Delta\alpha > c\left(1+2\lambda\right). \label{eq:PGG_dense_threshold_two_delta}
\end{equation}
Thus, the threshold $\Delta\alpha>3c$ corresponds to the symmetric case $\lambda=1$.
\subsection{\label{appendixE-part2}Stag-hunt game}
For the stag-hunt game, we obtain
\begin{align*}
G_A(y)&=\frac{(1+(k-2)y)^k}{4(k+1)(k-1)^{k-1}}\,b-\frac{c}{4},\\
G_I(y)&=\frac{(k+2)(1+(k-2)y)^k}{2k(k+1)(k-1)^{k-1}}\,b-\frac{c}{2}.
\end{align*}
Substituting these expressions in \eqref{eq:rho_two_delta}, the fixation probabilities will become
\begin{subequations}
\begin{align}
\rho_C&\approx\frac1N+\frac{\delta_A}{4}\left[\Delta\alpha+b\Phi_{C,A}^{\mathrm{SH}}(k)-c\right]+\frac{\delta_I}{4}\left[b\Phi_{C,I}^{\mathrm{SH}}(k)-2c\right],\label{eq:SH_rhoC_two_delta}\\
\rho_D&\approx\frac1N-\frac{\delta_A}{4}\left[\Delta\alpha+b\Phi_{D,A}^{\mathrm{SH}}(k)-c\right]-\frac{\delta_I}{4}\left[b\Phi_{D,I}^{\mathrm{SH}}(k)-2c\right],\label{eq:SH_rhoD_two_delta}\\
\frac{\rho_C}{\rho_D}&\approx1+\frac{N}{2}\Big\{\delta_A\left[\Delta\alpha+b\Phi_{R,A}^{\mathrm{SH}}(k)-c\right]+\delta_I\left[b\Phi_{R,I}^{\mathrm{SH}}(k)-2c\right]
\Big\},\label{eq:SH_ratio_two_delta}
\end{align}
\end{subequations}
where 
\begin{equation}
\Phi_{S,A}^{\mathrm{SH}}(k)=\frac{k}{3k+4}\Phi_S^{\mathrm{SH}}(k),\qquad\Phi_{S,I}^{\mathrm{SH}}(k)=\frac{2(k+2)}{3k+4}\Phi_S^{\mathrm{SH}}(k).\label{eq:SH_split_phi}
\end{equation}
for each \(S\in\{C,D,R\}\). The corresponding success condition has the form
\begin{equation}
\Delta\alpha+b\frac{k+2\lambda(k+2)}{3k+4}\Phi_S^{\mathrm{SH}}(k)>c(1+2\lambda).\label{eq:SH_condition_two_delta}
\end{equation}
Since all stag-hunt benefit coefficients vanish as \(k\to\infty\), large-degree conditions are equivalent and given by 
\begin{equation}
\Delta\alpha>c(1+2\lambda).\label{eq:SH_dense_threshold_two_delta}
\end{equation}
This again gives \(\Delta\alpha>3c\) when \(\lambda=1\).

If \(\Delta\alpha/c<1+2\lambda\), the critical benefit-to-cost thresholds are
\begin{equation}
\left(\frac{b}{c}\right)_{S}^{*}=\frac{(3k+4)(1+2\lambda-\Delta\alpha/c)}{(k+2\lambda(k+2))\Phi_S^{\mathrm{SH}}(k)},\label{eq:SH_threshold_two_delta}
\end{equation}
for $S\in\{C,D,R\}$.
Using the large-\(k\) asymptotics of the one-parameter coefficients, we obtain
\begin{align}
\left(\frac{b}{c}\right)_{\rho_C>1/N}^{*}&\sim\frac{1+2\lambda-\Delta\alpha/c}{2(1+2\lambda)}k^2,\label{eq:SH_asymp_C_two_delta}\\
\left(\frac{b}{c}\right)_{\rho_D<1/N}^{*}&\sim\frac{1+2\lambda-\Delta\alpha/c}{2(1+2\lambda)}k,\label{eq:SH_asymp_D_two_delta}\\
\left(\frac{b}{c}\right)_{\rho_C>\rho_D}^{*}&\sim\frac{1+2\lambda-\Delta\alpha/c}{1+2\lambda}k.\label{eq:SH_asymp_R_two_delta}
\end{align}
For \(\lambda=1\), these reduce to
\[
(3-\Delta\alpha/c)\frac{k^2}{6},\qquad(3-\Delta\alpha/c)\frac{k}{6},\qquad(3-\Delta\alpha/c)\frac{k}{3},
\]
as in the one-parameter model.

\subsection{\label{appendixE-part3}Snowdrift game}
Now, consider the snowdrift game.  After substituting the expectation identities \eqref{eq:Elplus1_inv} and \eqref{snow-eq1} into \(G_A\) and \(G_I\) given by \eqref{GA} and \eqref{GI}, we obtain the decompositions
\begin{equation}
G_A(y)=G_{b,A}(y)b-G_{c,A}(y)c
\end{equation}
and
\begin{equation}
G_I(y)=G_{b,I}(y)b-G_{c,I}(y)c.
\end{equation}
The benefit coefficients are
\begin{equation*}
\begin{aligned}
G_{b,A}(y)&=\frac{k-1}{4(k+1)}\left(\frac{k-1-(k-2)y}{k-1}\right)^k,\\
G_{b,I}(y)&=\frac{(k+2)(k-1)}{2k(k+1)}\left(\frac{k-1-(k-2)y}{k-1}\right)^k.
\end{aligned}
\end{equation*}
while the cost coefficients are
\begin{equation*}
\begin{aligned}
G_{c,A}(y)=&\frac{(k-1)\left[(k^2-1)+(k-2)(k^2+1)y\right]}{4k(k-2)y(k+1)^2\left[1+(k-2)y\right]}-\frac{(k-1)^2}{4k(k+1)(k-2)y}\left(1-\frac{k-2}{k-1}y\right)^{k+1}\\
&-\frac{(k-1)\left[3+(k-2)y\right]}{4k(k+1)^2\left[1+(k-2)y\right]}\left(\frac{(k-2)(1-y)}{k-1}\right)^k,\\
G_{c,I}(y)=&\frac{(k-1)\left[(k+1)(k^2-k-1)+(k-2)(k-1)(k^2+k-1)y\right]}{2k^2(k-2)y(k+1)^2\left[1+(k-2)y\right]}\\
&-\frac{(k-1)\left[3+(k-2)y\right]}{2k^2(k+1)^2\left[1+(k-2)y\right]}\left(\frac{(k-2)(1-y)}{k-1}\right)^k\\
&-\frac{(k-1)\left[k^2-k-1-(k-2)(k+1)y\right]}{2k^2(k+1)(k-2)y}\left(1-\frac{k-2}{k-1}y\right)^k.
\end{aligned}
\end{equation*}
Substituting these expressions in \eqref{eq:rho_two_delta}, the fixation probabilities will become
\begin{subequations}
\begin{align}
\rho_C&\approx\frac1N+\frac{\delta_A}{4}\left[\Delta\alpha+b\Phi_{C,A}^{\mathrm{SD}}(k)-c\Gamma^{\mathrm{SD}}_{C,A}\right]+\frac{\delta_I}{4}\left[b\Phi_{C,I}^{\mathrm{SD}}(k)-c\Gamma^{\mathrm{SD}}_{C,I}\right],\label{eq:SD_rhoC_two_delta}\\
\rho_D&\approx\frac1N-\frac{\delta_A}{4}\left[\Delta\alpha+b\Phi_{D,A}^{\mathrm{SD}}(k)-c\Gamma^{\mathrm{SD}}_{D,A}\right]-\frac{\delta_I}{4}\left[b\Phi_{D,I}^{\mathrm{SD}}(k)-c\Gamma^{\mathrm{SD}}_{D,I}\right],\label{eq:SD_rhoD_two_delta}\\
\frac{\rho_C}{\rho_D}&\approx1+\frac{N}{2}\Big\{\delta_A\left[\Delta\alpha+b\Phi_{R,A}^{\mathrm{SD}}(k)-c\Gamma^{\mathrm{SD}}_{R,A}\right]+\delta_I\left[b\Phi_{R,I}^{\mathrm{SD}}(k)-c\Gamma^{\mathrm{SD}}_{R,I}\right]
\Big\},\label{eq:SD_ratio_two_delta}
\end{align}
\end{subequations}
where 
\begin{equation}
\Phi_{S,A}^{\mathrm{SD}}(k)=\frac{k}{3k+4}\Phi_S^{\mathrm{SD}}(k),\qquad\Phi_{S,I}^{\mathrm{SD}}(k)=\frac{2(k+2)}{3k+4}\Phi_S^{\mathrm{SD}}(k).\label{eq:SD_split_phi}
\end{equation}
for each \(S\in\{C,D,R\}\).

The split cost coefficients decompose each original snowdrift cost coefficient
\(\Gamma_S^{\mathrm{SD}}(k)\), with \(S\in\{C,D,R\}\), into an aspiration part and an imitation part:
\begin{equation}
\Gamma_S^{\mathrm{SD}}(k)
=
\Gamma_{S,A}^{\mathrm{SD}}(k)
+
\Gamma_{S,I}^{\mathrm{SD}}(k).
\end{equation}
Equivalently, we write
\begin{equation}
\Gamma_{S,A}^{\mathrm{SD}}(k)
=
\theta_{S,A}^{\mathrm{SD}}(k)\Gamma_S^{\mathrm{SD}}(k),
\qquad
\Gamma_{S,I}^{\mathrm{SD}}(k)
=
\theta_{S,I}^{\mathrm{SD}}(k)\Gamma_S^{\mathrm{SD}}(k),
\end{equation}
where
\[
0\le \theta_{S,A}^{\mathrm{SD}}(k),\theta_{S,I}^{\mathrm{SD}}(k)\le 1,
\qquad
\theta_{S,A}^{\mathrm{SD}}(k)+\theta_{S,I}^{\mathrm{SD}}(k)=1.
\]
Unlike the benefit coefficients, these weights depend on \(S\) in general. In particular,
\(\theta_{C,A}^{\mathrm{SD}}(k)\), \(\theta_{D,A}^{\mathrm{SD}}(k)\), and
\(\theta_{R,A}^{\mathrm{SD}}(k)\) need not be equal for finite \(k\).

Under symmetric aspirations, \(\Delta\alpha=0\), the critical benefit-to-cost thresholds become
\begin{align*}
\left(\frac{b}{c}\right)_{\rho_C>1/N}^{*}&=\frac{\Gamma_{C,A}^{\mathrm{SD}}(k)+\lambda\Gamma_{C,I}^{\mathrm{SD}}(k)}{\Phi_{C,A}^{\mathrm{SD}}(k)+\lambda\Phi_{C,I}^{\mathrm{SD}}(k)},\\
\left(\frac{b}{c}\right)_{\rho_D<1/N}^{*}&=\frac{\Gamma_{D,A}^{\mathrm{SD}}(k)+\lambda\Gamma_{D,I}^{\mathrm{SD}}(k)}{\Phi_{D,A}^{\mathrm{SD}}(k)+\lambda\Phi_{D,I}^{\mathrm{SD}}(k)},\\
\left(\frac{b}{c}\right)_{\rho_C>\rho_D}^{*}&=\frac{\Gamma_{R,A}^{\mathrm{SD}}(k)+\lambda\Gamma_{R,I}^{\mathrm{SD}}(k)}{\Phi_{R,A}^{\mathrm{SD}}(k)+\lambda\Phi_{R,I}^{\mathrm{SD}}(k)}.
\end{align*}

The large-degree conclusion is unchanged. In the snowdrift game, both the benefit and cost weights vanish as \(k\to\infty\). Hence,
\begin{equation}
\rho_C^{\mathrm{SD}}\approx\frac1N+\frac{\delta_A\Delta\alpha}{4},
\qquad
\rho_D^{\mathrm{SD}}\approx\frac1N-\frac{\delta_A\Delta\alpha}{4},
\qquad
k\to\infty.
\label{eq:SD_dense_two_delta}
\end{equation}
Thus, for \(\delta_A>0\), the large-degree success condition remains
\begin{equation}
\Delta\alpha>0.
\label{eq:SD_dense_condition_two_delta}
\end{equation}

\subsection{Summary}
In summary, allowing aspiration responsiveness and payoff-biased
imitation to have different intensities changes the finite-degree
success conditions through the ratio
\(\lambda=\delta_I/\delta_A\), but it preserves the qualitative
separation between degree-dependent payoff effects and the
degree-independent aspiration gap. In the public goods and stag-hunt
games, the benefit contributions vanish in the sparse large-\(k\)
limit while residual cost terms remain, leading to the condition
\[
\Delta\alpha>c(1+2\lambda).
\]
In the snowdrift game, both the benefit and cost contributions vanish
in this limit, so the success condition reduces to
\[
\Delta\alpha>0
\]
when \(\delta_A>0\). Thus, the symmetric model
\(\delta_A=\delta_I\) corresponds to the special case \(\lambda=1\),
whereas the dilution mechanism itself does not depend on this symmetry.


\bibliographystyle{unsrt}

\bibliography{bibliography}

\end{document}